\documentclass{sig-alternate}
\usepackage{rotating}
\usepackage{url}
\usepackage{algorithmic}
\usepackage{textcomp}

\newtheorem{theorem}{Theorem}
\newtheorem{lemma}[theorem]{Lemma}

\newtheorem{corollary}[theorem]{Corollary}

\DeclareMathOperator{\sigSym}{\mathfrak{s}}
\DeclareMathOperator{\headSym}{in}
\DeclareMathOperator{\lcm}{lcm}
\newcommand{\mmorSym}{\phi}

\newcommand{\comment}[1]{}

\newcommand{\hideAppendix}[1]{#1}\newcommand{\appRef}[1]{Appendix \ref{#1}}\newcommand{\AppRef}[1]{Appendix \ref{#1}}

\newcommand{\repthm}[3]
{{
  \renewcommand{\thetheorem}{\ref{#2}}
  \begin{#1}
  {#3}
  \end{#1}
  \addtocounter{theorem}{-1}
}}

\newcommand{\p}{\ensuremath{^\prime}}

\newcommand{\biimp}{\Leftrightarrow}

\newcommand{\proj}[1]{\overline{#1}}
\newcommand{\mbasis}[1]{\boldsymbol{e}_{#1}}
\newcommand{\sig}[1]{\sigSym\left({#1}\right)}
\newcommand{\origBasis}{\mathcal G}
\newcommand{\extendBasis}[1]{\origBasis_{#1}}
\newcommand{\mmor}[1]{\mmorSym\left({#1}\right)}

\newcommand{\set}[1]{\left\{{#1}\right\}}
\newcommand{\setBuilder}[2]{\left\{{#1}\left|{#2}\right.\right\}}

\newcommand{\ideal}[1]{\left\langle{#1}\right\rangle}
\newcommand{\hd}[1]{\headSym({#1})}
\newcommand{\hdp}[1]{\hd{\proj{#1}}}

\newcommand{\spair}[2]{\ensuremath{\mathcal{S}\left({#1},{#2}\right)}}
\newcommand{\spoly}[2]{\ensuremath{\mathcal{P}\left({#1},{#2}\right)}}
\newcommand{\koz}[2]{\ensuremath{\mathcal{K}\left({#1},{#2}\right)}}
\newcommand{\ratid}[1]{\ensuremath{\tau_{#1}}}

\newcommand*{\defeq}{\mathrel{\vcenter{\baselineskip0.5ex \lineskiplimit0pt
                     \hbox{\scriptsize.}\hbox{\scriptsize.}}}%
                     =}

\newcommand{\sigdiv}{\ensuremath{\sigSym{}}-division}
\newcommand{\sigdive}{\ensuremath{\sigSym{}}-divide}

\newcommand{\sigred}{\ensuremath{\sigSym{}}-reduction}
\newcommand{\sigrede}{\ensuremath{\sigSym{}}-reduce}

\newcommand{\grobner}{Gr\"obner}
\newcommand{\galg}{SB}
\newcommand{\galgex}{the SB algorithm}

\newcommand{\proofPart}[1]{{\bf $\boldsymbol{#1}$:}}

\begin{document}
\conferenceinfo{ISSAC}{2012 Grenoble, France}
\title{Practical Gr\"obner Basis Computation}

\numberofauthors{2}
\author{
\alignauthor
Bjarke Hammersholt Roune\titlenote{Supported by The Danish Council for Independent Research | Natural Sciences.}\\
  \affaddr{Cornell University}\\
  \affaddr{Department of Mathematics}\\
  \affaddr{Ithaca 14853-4201, USA}\\
  \email{www.broune.com}
 \alignauthor
Michael Stillman\titlenote{Partially supported by NSF grant DMS-1002210}\\
  \affaddr{Cornell University}\\
  \affaddr{Department of Mathematics}\\
  \affaddr{Ithaca 14853-4201, USA}\\
  \email{mike@math.cornell.edu}
}

\maketitle
\begin{abstract}
We report on our experiences exploring state of the art Gr\"obner
basis computation. We investigate signature based algorithms in
detail. We also introduce new practical data structures and
computational techniques for use in both signature based Gr\"obner
basis algorithms and more traditional variations of the classic
Buchberger algorithm. Our conclusions are based on experiments using
our new freely available open source standalone C++ library.
\end{abstract}

\section{Introduction}

Since F5 \cite{f5}, there have been several signatured-based
algorithms. Recently Gao, Volny and Wang (GVW) \cite{gvwSyzygyF5} have
introduced a signature algorithm that generalizes several previous
algorithms. Arri and Perry (AP) \cite{apSyzygyF5} have published a
very similar algorithm. We will refer to both algorithms as
\emph{\galg{}} for Signature Basis algorithm.

We give another way of describing \galg{} (Section \ref{sec:alg}) and
we employ a rewriting criterion that improves on that of GVW. We show
how this rewriting criterion can be used to eliminate S-pairs (Section
\ref{sec:termChoice}). This enables us to characterize the number of
S-pair reductions that \galg{} performs in terms of the final basis
and the initial submodule of the module of syzygies (Theorem
\ref{thm:characReductions}).

We introduce new practical data structures and computational
techniques for use in both signature based Gr\"obner basis algorithms
and more traditional variations of the classic Buchberger algorithm
(sections \ref{sec:sigImp} and  \ref{sec:ds}).

Our ex\-pe\-ri\-ments are based on our new freely available open
source standalone C++ library (Section \ref{sec:exp}) \cite{onlinePaper}. 
The library was
written with the intention to use it in Macaulay 2, but it does not
depend on any component of Macaulay 2 and could be used by any other
system as well. It is also possible to take just the data structures
from the library and use them in another codebase. We welcome all to
read, use and modify the code. We are happy to receive suggestions
for improvements.

Due to the page limit the paper is of necessity compressed. {\bf The
  proofs and other supporting material are available in appendices
  that do not appear in the printed paper}, but that are available
online \cite{onlinePaper}.

\section{The \galg{} Algorithm}
\label{sec:alg}

In this section we introduce a simplest possible version of \galg{}
from first principles. The description differs from both GVW and AP
though the results are not new. See \appRef{app:alg} for proofs.

\subsection{Notation and Terminology}
\label{sec:algsetup}

Let $R$ be a polynomial ring over a field. Let $\origBasis$ be a
finite set of non-zero polynomials of $R$ indexed as
$g_1,\ldots,g_m$. Consider the free module $R^m$ and define the
homomorphism $u\mapsto\proj u\colon R^m\rightarrow R$ by $\proj
u\defeq\sum_{i=1}^mu_ig_i.$ We say that $u$ is a \emph{representation}
of $\proj u$. Then by definition
$\ideal{\origBasis}\defeq\proj{R^m}=\setBuilder{\bar u}{u\in R^m}$.

Let $\mbasis 1,\ldots,\mbasis m$ be the standard basis of unit vectors
in $R^m$. Let $\leq$ denote two different term orders -- one on $R^m$
and one on $R$. We require these two orders to be related such that
$a\leq b\biimp a\mbasis i\leq b\mbasis i$ for all monomials $a,b$ and
$i=1,\ldots,n$. For both orders we use the convention that $0<t$ for
all terms $t$.

Let the \emph{signature} $\sig u$ be the $\leq$-maximal term of $u\in
R^m$ and define $\sig 0=0$. Let the \emph{lead term} $\hd f$ be the
$\leq$-maximal term of $f\in R$ and define $\hd{0}=0$. In this way
every $u\in R^m$ has two main associated characteristics -- the
signature $\sig u$ and the lead term $\hdp u$ of its image in $R$.

We will consider extensions $\extendBasis n\supseteq\origBasis$ of
additional non-zero polynomials from $\ideal{\origBasis}$ indexed as
$
g_1,\ldots,g_m,g_{m+1},\ldots,g_n
$
where $m\leq n$. Define $v\mapsto\proj v$, the basis $\mbasis
1,\ldots,\mbasis n$ and representation relative to $\extendBasis n$ as
above. Each extension $\extendBasis n$ will come with a homomorphism
$\mmorSym\colon R^n\rightarrow R^m$ such that $\proj{u}=\proj{\mmor
  u}$. Then in particular $\mmor{\mbasis i}$ for $i>m$ is a
representation of $g_i$ in terms of $g_1,\ldots,g_m$. We extend the
definition of signature from $R^m$ to $R^n$ by $\sig u\defeq
\sig{\mmor u}$. A \emph{syzygy} is a $u\in R^n$ or $u\in R^m$ such
that $\proj u=0$. None of the $\mbasis i$ are syzygies as
$\proj{\mbasis i}=g_i\neq 0$.

Note that the signature of $u\in R^n$ depends on $\mmorSym$ even if
that is not apparent from the notation $\sig u$. An element $f\in R$
can have many different representations in $R^n$ with distinct
signatures. We write $a\simeq b$ for two terms $a$ and $b$ if $a=sb$
where $s$ is a non-zero element of the ground field. So $a\leq b\leq
a$ if and only if $a\simeq b$.

A \emph{monomial} is a polynomial with exactly one term. A monomial or
term with a coefficient of 1 is \emph{monic}. Neither terms nor
monomials are necessarily monic. In particular $\sig u$ and $\hdp u$
for $u\in R^n$ are not necessarily monic.

\subsection{Division With Signatures}

\galg{} uses a notion of division called \emph{\sigdiv{}}. It is
similar to classic polynomial division. The main difference is that we
start with an element $u\in R^n$ instead of an element of $R$ and we
take signature into account. The result of division of $u\in R^n$ by
$\extendBasis n$ is a quotient $q\in R^n$ and a remainder $r\in R^n$
such that
\begin{enumerate}
\item $u=q+r$,
\item $\hdp u\geq\hd{q_ig_i}\text{ for
}i=1,\ldots,n$,
\item if $\sig u\geq\sig{a\mbasis i}$ for a monomial $a$
then $\hdp{a\mbasis i}$ does not equal any term of $\proj r$,
\item $\sig u\geq\sig q$.
\end{enumerate}

We \emph{\sigdive{}} to get the quotient $q$ and we \emph{\sigrede{}}
to get the remainder $r$. $u$ is \emph{reduced} if $q=0$ and otherwise
$u$ is \emph{reducible}. We say that $u$ \emph{reduces to zero} if
$\proj r=0$.

The first three conditions are similar to conditions on the quotient
and remainder in classic polynomial division. The fourth condition
disallows some reduction steps. Without signatures the condition for
$g_i$ to reduce a term $t$ of $f$ is just that $\hd{g_i}|t$. We would
then determine the monomial $a$ such that $\hd{ag_i}=t$ and reduce to
$f-ag_i$.

With signatures it is not sufficient that $\hd{g_i}|t$ in order for
$\mbasis i$ to reduce a term $t$ of $\proj u$. Here $\mbasis i$
reduces $t$ when there is a term $a$ such that $\hdp{a\mbasis i}=t$
and $\sig{a\mbasis i}\leq\sig u$. The outcome of such a reduction step
is then $u-a\mbasis i$. So a reduction step happens when it can be
carried out without strictly increasing the signature. A reduction
step is \emph{singular} if $\sig{a\mbasis i}\simeq\sig u$. When
$\mbasis i$ reduces $t$ it is convenient to also say that $a\mbasis i$
reduces $t$ so that $a$ is introduced right then.

In \emph{top \sigred{}} the reduction stops when the lead term of
$\proj u$ cannot be reduced. We say that $u$ is \emph{top reducible}
if the lead term of $\proj u$ can be reduced and otherwise $u$ is
\emph{top reduced}.

In \sigdiv{} the signature is not allowed to increase. \emph{Regular
  division} is like \sigdiv{} except that the coefficient of the
signature is not allowed to change either. So the difference is that
singular reduction steps are not allowed.

A basis $\extendBasis n$ is a \emph{signature \grobner{} basis} if all
$u\in R^n$ \sigrede{} to zero. Then $\extendBasis n$ maps to a
\grobner{} basis $\proj{\extendBasis n}$ as then all elements of
$\ideal{\proj{\extendBasis n}}=\proj{R^n}$ reduce to zero.

\subsection{S-pairs}
\label{sec:spairs}

The \emph{S-pair} between $\mbasis i$ and $\mbasis j$ is $\spair i
j\defeq \spair {\mbasis i} {\mbasis j}\defeq
\frac{\hd{g_j}}{d}\mbasis i-
\frac{\hd{g_i}}{d}\mbasis j$
where $d\defeq \gcd\left(\hd{g_i},\hd{g_j}\right)$.
If $\sig{a\mbasis i}\simeq\sig{b\mbasis j}$ then the S-pair is
\emph{singular} and otherwise it is \emph{regular}. By ``S-pair'' we
always mean ``regular S-pair''.

\galg{} reduces S-pairs using regular reduction steps and adds the
regular reduced result to the basis if it is not a syzygy and not
singular top reducible. Theorem \ref{thm:spairs} implies that when all
S-pairs have been processed in this fashion, then the basis is a
signature \grobner{} basis.

\begin{theorem}
\label{thm:spairs}
Let $T$ be a term of $R^m$. Assume for all S-pairs $p$ with $\sig
p\leq T$ that if $p\p$ is the result of regular reducing $p$, then
$p\p$ is singular top reducible or a syzygy. Then all elements $u\in
R^n$ with $\sig u\leq T$ \sigrede{} to zero.
\end{theorem}

The outcome of polynomial reduction depends on the choice of reducer,
so the choice of reducer can change what the intermediate basis is in
the classic Buchberger algorithm. In contrast to this, Lemma
\ref{lem:p3} implies that the regular reduction of S-pairs has a
uniquely determined remainder.

\begin{lemma}
\label{lem:p3}
Let $L\in R^m$ be a term such that all $v\in R^n$ with $\sig v<L$
\sigrede{} to zero. Let $a,b\in R^n$ such that $\sig a=\sig b\leq
L$. Then $\hdp a=\hdp b$ if $a$ and $b$ are regular top reduced. Also,
$\proj a=\proj b$ if $a$ and $b$ are regular reduced.
\end{lemma}

Theorem \ref{thm:syzygies} implies that the S-pairs of a signature
\grobner{} basis give rise to a \grobner{} basis of the syzygy module
of $\origBasis$. Note that we are talking about the syzygy module of
the \emph{original} basis $\origBasis$ rather than the syzygy module
of a \grobner{} basis of the same ideal. The former is generally much
harder to compute than the latter. GVW present this result, while AP
essentially prove it but do not state it.

\begin{theorem}
\label{thm:syzygies}
Let $\extendBasis n$ be a signature \grobner{} basis and let $u\in
R^m$ be a syzygy. Then there is an S-pair $p$ that regular reduces to
a syzygy $p\p$ such that $\sig{p\p}$ divides $\sig u$.
\end{theorem}

\galg{} is known to terminate --- see \appRef{app:term}.

\subsection{Pseudo code}
\label{sec:pseudo}

Here is pseudo code for a \emph{simplest possible version} of
\galg{}. This pseudo code should not be taken as a guide to efficient
implementation. The code computes a signature \grobner{} basis
$\extendBasis n$ and the initial submodule $H$ of the syzygy module of
$g_1,\ldots,g_m$. An actual implementation would keep track of monic
pairs $(h_i,s_i)$ where $h_i\defeq \proj{\mbasis i}$ and
$s_i\defeq\sig{\mbasis i}$ instead of maintaining a full
representation $\mmor{\mbasis i}$ of each $g_i$.  \ \\ \ \\ {\bf
  SignatureBuchberger($\set{g_1,\ldots,g_m}\subseteq R$)}
\begin{algorithmic}
\STATE $n\gets m$
\STATE $S\gets \setBuilder{\spair i j}{1\leq i<j\leq m\text{ and $\spair i j$ is regular}}$
\STATE $H\gets \ideal 0\subseteq R^m$
\WHILE{$S\neq\emptyset$}
  \STATE $p\gets$ an element of $S$ with $\leq$-minimal signature
  \STATE $S\gets S\setminus\set p$
  \STATE $p\p\gets$ result of regular reducing $p$
  \IF {$\proj{p\p}=0$}
    \STATE $H\gets H+\ideal{\sig{p\p}}$
  \ELSIF {$p\p$ is not singular top reducible}
    \STATE $n\gets n + 1$
    \STATE $\phi(\mbasis n)\gets \phi(p\p)$\quad
      \COMMENT{implies $g_n=\proj{p\p}$ and $\sig{\mbasis n}=\sig{p\p}$}
    \STATE $S\gets S\cup\setBuilder{\spair i n}{i < n\text{ and $\spair i j$ is regular}}$
  \ENDIF
\ENDWHILE
\end{algorithmic}

\section{Signature Alg. Improvements}
\label{sec:sigImp}

In this section we show techniques that improve signature \grobner{}
basis computation. Several of these techniques apply to signature
algorithms in general rather than just \galg{}.

\subsection{S-pair Elimination}
\label{sec:spairElim}

There are many S-pairs that it is not necessary for \galg{} to reduce
in order to arrive at a signature \grobner{} basis. We say that we
\emph{eliminate} an S-pair when we determine that it is not necessary
to reduce that S-pair. The S-pair elimination criteria that we present
here in Section \ref{sec:spairElim} are already present in both GVW
and AP.

Three things can happen when \galg{} regular reduces an S-pair in
signature $T$ and gets a remainder $r$. First, $r$ might be a syzygy
in which case its signature is added to the set of known syzygy
signatures. Second, $r$ might be singular top reducible in which case
$r$ is thrown away. Third, if $r$ is not a syzygy and not singular top
reducible, then $r$ is added to the basis. For these three cases we
say respectively that $T$ is a \emph{syzygy}, \emph{singular} or
\emph{basis} signature. These notions are well defined since Lemma
\ref{lem:p3} implies that regular reduction of S-pairs yields a
uniquely determined remainder.

Recall from the definition of S-pair that any mention in this paper of
S-pairs refers to regular S-pairs. So \galg{} can right away eliminate
any S-pair that is not regular.

If there is more than one S-pair in signature $T$, then we only have
to reduce one of them since Lemma \ref{lem:p3} implies that they will
all have the same remainder upon regular reduction. Section
\ref{sec:termChoice} develops this topic further.

Suppose that $T$ is an S-pair signature and we know of a syzygy
signature $L$ that divides $T$. Then $T$ is a syzygy signature by
Corollary \ref{col:spairSyzygy}, which allows us to eliminate all
S-pairs in signature T. Call this the \emph{signature
  criterion}. Since \galg{} considers S-pairs in ascending order of
signature, we observe that the only syzygy signatures that the
signature criterion might not eliminate are those that come from an
element of a minimal \grobner{} basis of the module of syzygies.

\begin{corollary}
\label{col:spairSyzygy}
Let $u\in R^n$ such that all $v\in R^n$ with $\sig v<\sig u$ reduce to
zero. Suppose there exists a syzygy $h\in R^n$ whose signature divides
the signature of $u$. Then $u$ regular reduces to zero.
\end{corollary}

Table \ref{tab:sigSPair} gives information about how many S-pairs are
eliminated due to each criterion. The criteria are checked in the
given order. \AppRef{app:spairElim} has more details.

\subsection{Rewriting and the Singular Criterion}
\label{sec:termChoice}

We present a technique that makes reductions easier to perform and
that provides a criterion for eliminating all S-pairs of singular
signature.

If there are two or more S-pairs in the same signature $T$, then we
only have to regular reduce one of them. Since reduction proceeds by
decreasing the lead term, we can heuristically speed up reduction by
choosing an S-pair $p$ whose lead term $\hdp p$ is minimal. Both GVW
and AP make this suggestion. In F5, the situation where one S-pair is
preferred over another is called a \emph{rewriting criterion}.

Selecting the minimum lead term S-pair would require us to calculate
the lead term of each S-pair. If
$\sig{\spair\alpha\beta}=\sig{t\alpha}$, then we get the same result
from regular reducing $\spair\alpha\beta$ as for regular reducing
$t\alpha$. So we should select the term from $M$ with minimal lead
term, where $M\defeq\setBuilder{t\mbasis i}{t\text{ is a monomial and
  }\sig{t\mbasis i}=T}$. Let $t\mbasis i$ be an element of $M$ with
minimal lead term. Note that $t\mbasis i$ might not come from any
S-pair in signature $T$. If $t\mbasis i$ is not regular top reducible,
then we know that $T$ is a singular signature, so no regular reduction
need take place. We call this criterion for eliminating S-pairs the
\emph{singular criterion}. Corollary \ref{cor:p4} implies that the
singular criterion eliminates an S-pair if and only if it is of
singular signature.

AP independently came up with the idea of the singular criterion, and
they additionally remark that if $t\mbasis i$ does not come from an
S-pair in signature $T$, and if $t\mbasis i$ has a strictly lower lead
term than any element of $M$ that does come from an S-pair in
signature $T$, then the singular criterion will apply. An implication
of that is that if $\sig{\spair\alpha\beta}=\sig{q\alpha}$ and there
exists a $w\mbasis k$ such that $\sig{w\mbasis k}=\sig{q\alpha}$ and
$\hdp{w\mbasis k}<\hdp{q\alpha}$ then we can eliminate
$\spair\alpha\beta$ right away.

\begin{corollary}
\label{cor:p4}
Let $p$ be an S-pair and let $p\p$ be the result of regular reducing
$p$. Let $M$ be the finite set
\[
M\defeq\setBuilder{a\mbasis i}{a\text{ is a monomial and }\sig{a\mbasis i}=\sig p}.
\]
Then all elements of $M$ regular reduce to $p\p$. Also, $p\p$ is
singular top reducible if and only if some element of $M$ is regular
top reduced.
\end{corollary}

\subsection{Koszul Syzygies for S-Pair Elimination}

The \emph{Koszul syzygy} between $\mbasis i$ and $\mbasis j$ is $\koz
i j\defeq g_j\mbasis i-g_i\mbasis j$. If $\sig{g_j\mbasis
  i}\not\simeq\sig{g_i\mbasis j}$ then the Koszul syzygy is
\emph{regular}. By ``Koszul syzygy'' we always mean ``regular
Koszul syzygy''.

The signature of $\koz i j$ is $\max\left(\hd{g_i}\sig{\mbasis
  j},\hd{g_j}\sig{\mbasis i}\right)$.  We can use the Koszul syzygy to
eliminate all S-pairs in that signature. We call this S-pair
elimination criterion the \emph{Koszul criterion}. The idea goes back
to at least F5.

In the classic Buchberger algorithm, we can eliminate an S-pair
between two polynomials if their lead terms are relatively prime. In
\galg{} this is a special case of the Koszul criterion, since then
$\sig{\koz i j}=\sig{\spair i j}$. Even so, we call this the
\emph{relatively prime criterion} and do not consider such S-pairs to
be eliminated by the Koszul criterion.

Many S-pairs that could be eliminated by the Koszul criterion are
already eliminated by the signature criterion. We consider such
S-pairs to be eliminated by the signature criterion rather than the
Koszul criterion. The signature criterion eliminates all syzygy
signatures that are divisible by some other syzygy signature. So the
Koszul criterion eliminates those signatures that are the signature of
a Koszul syzygy, that are minimal among all syzygy signatures and that
are not eliminated by the relatively prime criterion.

A straight forward way to make use of Koszul syzygies would be to
insert all Koszul syzygy signatures into the set of known syzygies and
let the signature criterion also work with Koszul syzygies. Since
Koszul syzygy signatures are rarely minimal syzygy signatures, this
can greatly increase the size of the set of known syzygies. We can
remove the Koszul signatures that are non-minimal among the syzygy
signatures that we know at any given point, but many of the signatures
that are left after that will still not be minimal among the set of
all syzygy signatures. So adding the Koszul signatures to the set of
known syzygy signatures can cause significant overhead in space for
storing all these signatures, and significant overhead in time for
checking all these signatures when using the signature criterion.

An implication of the discussion so far is that the Koszul criterion
only eliminates an S-pair that would not already have been eliminated
by the signature criterion when the Koszul signature is equal to the
S-pair signature. This implies that we can make full use of the Koszul
criterion by maintaining a priority queue of Koszul signatures. The
priority queue allows us to determine the minimum Koszul signature in
the queue. S-pairs are processed in increasing order of signature, so
if we have gotten to an S-pair with signature $T$ and the minimum
Koszul signature $L$ is less than $T$, then we can throw $L$ away
since it will never equal any future S-pair signatures. If $T=L$ then
we can eliminate the S-pair using the Koszul criterion. If $T<L$ then
the Koszul criterion cannot eliminate the S-pair.

The priority queue approach does not alleviate the memory overhead of
Koszul syzygies. We still have to construct all the Koszul
signatures, which by itself can take a lot of time. One observation
that helps is that if the S-pair $\spair i j$ is eliminated by the
signature criterion using a known syzygy signature $T$, then
$T|\sig{\spair i j}|\sig{\koz i j}$. So if $\spair i j$ is eliminated
by the signature criterion, then there is no reason to construct the
Koszul signature $\sig{\koz i j}$.

Another observation is that $\sig{\spair i j}\leq\sig{\koz i
  j}$. S-pairs are processed in order of increasing signature, so if
$\koz i j$ ends up eliminating an S-pair, then that S-pair must have a
higher signature than $\spair i j$ does. So we do not need to
construct $\koz i j$ until we process the S-pair $\spair i j$. If the
S-pair is eliminated or reduces to zero, then we do not have to
construct the Koszul syzygy at all. These two observations reduce the
overhead in time and space of the Koszul criterion to almost nothing.

We can now characterize how many S-pairs \galg{} reduces when using
these S-pair elimination techniques.

\begin{theorem}
\label{thm:characReductions}
Let $\extendBasis n$ be a minimal signature \grobner{} basis. Let $H$
be the initial submodule of the module of syzygies of
$\origBasis$. Then \galg{} reduces one S-pair for each element of
$\extendBasis n\setminus\origBasis$. \galg{} also reduces one S-pair
for each minimal generator of $H$ that is not the signature of any
Koszul syzygy among the elements of $\extendBasis n$. \galg{} reduces
no more S-pairs than that.
\end{theorem}

Table \ref{tab:sigSPair} shows that the relatively prime and Koszul
criteria eliminate only a small proportion of all S-pairs. However,
for several examples, that is still a significant proportion of the
amount of S-pairs that are reduced, so these criteria can have a
significant impact on the final number of reductions.

\subsection{Compare Ratios Instead of Signatures}
\label{sec:ratio}

\galg{} can spend a lot of time comparing signatures. We present a
technique for replacing many of these signature comparisons with
comparison of just a single integer.

We start with the example of computing the signature of an S-pair
$p\defeq\spair i j$. Let $A\defeq\sig{\mbasis i}$,
$a\defeq\hdp{\mbasis i}$, $B\defeq\sig{\mbasis j}$ and
$b\defeq\hdp{\mbasis j}$. Then $\sig p$ is the larger one of
$E\defeq\frac{b}{\gcd(a,b)}A$ and $F\defeq\frac{a}{\gcd(a,b)}B$, so a
straight forward way of getting $\sig p$ is to compute $E$ and $F$ and
then compare the two to see which is larger. Computing $\sig p$ in
this way can take a lot of time when there are many basis
elements. For a faster solution, observe that (we allow negative
exponents)
\[
\textstyle{
E<F \quad\biimp\quad
bA<aB \quad\biimp\quad
\frac{A}{a}<\frac{B}{b}.}
\]
So if we store the ratio of signature to lead term (the \emph{sig-lead
  ratio}) with each basis element, then we can determine which of $E$
and $F$ is larger by comparing the stored ratios instead of having to
compute both $E$ and $F$. The next step is then to compute only the
one out of $E$ and $F$ that is the signature. The signature of the
Koszul syzygy $\koz i j$ is the larger one of $bA$ and $aB$, so the
same technique works there.

We speed up the comparison of sig-lead ratios by associating an
integer $\ratid i$ to each basis element $\mbasis i$ such that $\ratid
i < \ratid j\biimp\frac{A}{a} < \frac{B}{b}$. In this way sig-lead
ratio comparisons can be done as an integer comparison, which is much
faster. To support fast insertions of new basis elements we use a
simple approach based on spacing the integers far apart to begin with
and just rebuilding the whole datastructure if a new basis element
would need to have an integer between $k$ and $k+1$ for some
$k$. There exists a faster approach \cite{maintainingOrder} for
updating the integers $\ratid i$, but we have not needed it.

Ratio comparisons also occur in \sigred{}. Suppose we want to
\sigrede{} the lead monomial of $u$ by a basis element $e_i$. Let
$A\defeq\sig u$, $a=\hdp u$, $B=\sig{\mbasis i}$ and $b=\hdp{\mbasis
  i}$. Suppose that $b|a$. Then the \sigred{} can be performed if
$\frac{a}{b}B\leq A$, which is equivalent to
$\frac{B}{b}\leq\frac{A}{a}$. Unfortunately, the sig-lead ratio
$\frac{A}{a}$ has to be computed for each term being reduced. In our
ex\-pe\-ri\-ments it has been slower to compute an appropriate integer
to associate to the ratio than to just use $\frac{A}{a}$ directly for
comparisons. This is still faster than deciding $\frac{a}{b}B<A$
directly since that expression involves a division and a
multiplication for every comparison.

Ratio comparisons are also relevant to finding the best module term to
reduce as described in Section \ref{sec:termChoice}. We need to find
the element of $M$ from Corollary \ref{cor:p4} with minimum lead
term. Each basis element $\alpha$ whose signature $A\defeq\sig\alpha$
divides $T$ contributes the element $A\p\defeq\frac{T}{A}\alpha$ to
$M$. The lead term is $\hdp{A\p}=\frac TAa$ where
$a\defeq\hdp\alpha$. If $\beta$ is another basis element that also
contributes an element $B\p$ to $M$, then we need to perform the lead
term comparison
\[
\textstyle{
\frac TAa<\frac TBb \quad\biimp\quad
\frac aA < \frac bB \quad\biimp\quad
\frac Aa > \frac Bb,
}
\]
which can be done as a sig-lead ratio comparison.

The sig-lead ratio comparison technique also applies to what we call
\emph{base divisors} --- see Section \ref{sec:basediv}.

\subsection{Base Divisors for S-Pair Elimination}
\label{sec:basediv}

When a new element is added to the basis, we construct the S-pairs
between the new element and all the previous elements. In many cases
almost all of these S-pairs can be eliminated right away by the
signature criterion, so if there are many S-pairs then just
constructing the signature for each S-pair and then checking the
signature criterion on that signature can take up a lot of time.

We present a new S-pair elimination criterion that we call the
\emph{base divisor criterion}. The new criterion is strictly weaker
than the signature criterion, but it is faster to check.

We check the base divisor criterion before the signature criterion, so
all the S-pairs that the base divisor criterion eliminates then do not
need to be checked by the slower signature criterion. Table
\ref{tab:sigSPair} shows that the base divisor criterion eliminates a
substantial amount of S-pairs. Table \ref{tab:sbTimes} shows the drop
in performance if we do not check the base divisor criterion before
using the signature criterion. For yang1 the base divisors are a 35\%
performance improvement, and they eliminate 71\% of the S-pairs.

Let $\beta$ be a new basis element that we have just added to the
basis. We consider the S-pairs $\spair\beta\gamma$ between $\beta$ and
each other basis element $\gamma$. We aim to eliminate
$\spair\beta\gamma$ without computing the signature
$\sig{\spair\beta\gamma}$.

The idea here is to consider a fixed previous basis element $\alpha$
that has certain properties. We call $\alpha$ a \emph{base
  divisor}. We would like it to be true that
$\sig{\spair\alpha\gamma}|\sig{\spair\beta\gamma}$, since then we can
eliminate $\spair\beta\gamma$ when $\spair\alpha\gamma$ has a syzygy
signature. We use a triangle of bits, one bit for each S-pair, so we
can see in constant time if we know $\spair\alpha\gamma$ to have a
syzygy signature.  The criterion has two parts depending on whether
$\gamma$ has a high or low sig-lead ratio (see Section
\ref{sec:ratio}).

\subsubsection*{High Ratio Base Divisors}

A basis element $\alpha$ can be used as a \emph{high ratio base
  divisor} when $\hdp\alpha|\hdp\beta$. A high ratio base divisor
$\alpha$ can eliminate $\spair\beta\gamma$ when $\gamma$ has higher
sig-lead ratio than both of $\alpha$ and $\beta$. Theorem
\ref{thm:highRatio} spells out the precise criterion.

\begin{theorem}
\label{thm:highRatio}
Let $\alpha,\beta,\gamma\in R^n$ such that $\hdp\alpha|\hdp\beta$ and
$\frac{\sig\gamma}{\hdp\gamma}> \frac{\sig\alpha}{\hdp\alpha},
\frac{\sig\beta}{\hdp\beta}$. Then
$\sig{\spair\alpha\gamma}|\sig{\spair\beta\gamma}$.
\end{theorem}

We can use a kd-tree (see Section \ref{sec:monds}) on the initial
terms of the basis elements to quickly determine all the possible base
divisors. The base divisor that will eliminate the most S-pairs is the
one with the highest sig-lead ratio, so we use that one. Sometimes
there is no high ratio base divisor.

\subsubsection*{Low Ratio Base Divisors}

A basis element $\alpha$ can be used as a \emph{low ratio base
  divisor} when $\sig\alpha|\sig\beta$. A low ratio base divisor
$\alpha$ can in some cases eliminate $\spair\beta\gamma$ when $\gamma$
has lower sig-lead ratio than both of $\alpha$ and $\beta$. Theorem
\ref{thm:lowRatio} spells out the precise criterion.

\begin{theorem}
\label{thm:lowRatio}
Let $\alpha,\beta,\gamma\in R^n$ such that $\sig\alpha|\sig\beta$ and
$ \frac{\sig\gamma}{\hdp\gamma}< \frac{\sig\alpha}{\hdp\alpha},
\frac{\sig\beta}{\hdp\beta}$. Let
$x^p\defeq\frac{\hdp\alpha\sig\beta}{\sig\alpha}$,
$x^a\defeq\hdp\alpha$ and $x^b\defeq\hdp\beta$. Define $v$ by
$v_i\defeq\infty$ for $b_i\leq p_i$ and $v_i\defeq\max(p_i,a_i)$
otherwise.  Then $\sig{\spair\alpha\gamma}|\sig{\spair\beta\gamma}$ if
and only if $\hdp\gamma|x^v$.
\end{theorem}

To use Theorem \ref{thm:lowRatio} to eliminate $\spair\beta\gamma$, we
have to check if $\hdp\gamma|x^v$. Since $v$ does not depend on
$\gamma$, we need only compute it once. We need not check that
$\hdp\gamma|x^v$ if
$\frac{\sig\alpha}{\hdp\alpha}|\frac{\sig\beta}{\hdp\beta}$ since then
$v_i=\infty$ for every entry, but such $\alpha$ are rare.

In order for the sig-lead ratio requirement to be satisfied as often
as possible, we choose an $\alpha$ with maximum sig-lead ratio. A
kd-tree (see Section \ref{sec:monds}) on the basis signatures can
quickly find all the possible low ratio base divisors.

The numbers in Table \ref{tab:sbTimes} and Table \ref{tab:reducer}
are based on two base divisors per $\beta$, as that minimized the
total runtime.

\section{Data structures}
\label{sec:ds}

We present data structures that are useful both for classic Buchberger
algorithms and for signature algorithms.

\subsection{Ordering Terms During Reduction}
\label{sec:prio}

Both polynomial reduction and \sigred{} operate by having a
\emph{current polynomial} $f$ and adding monomial multiples $mg_i$ to
$f$. The basic operations for keeping track of the terms of $f$ are to
extract the maximal remaining term of $f$ and to add polynomials of
the form $mg_i$ to $f$. So we need a \emph{priority queue} on the
terms of $f$. Adding elements to a priority queue is called a
\emph{push} while removing the maximal element is called a
\emph{pop}. We investigate priority queues for keeping track of the
terms of $f$.

One solution is to store $f$ directly as a polynomial whose terms are
sorted. Yan pointed out that this can be very slow, as we can end up
looking at every term of $f$ for every insertion even when $mg_i$ has
only two terms. Yan introduced the \emph{geobucket priority queue}
which alleviates this issue \cite{geobucket}.

Heaps are a popular priority queue. Monagan and Pearce present
experiments that indicate that heaps are better than geobuckets for
polynomial multiplication and division \cite{heapOfPointers}.

The priority queue on terms of $f$ can contain terms $a$ and $b$ such
that $a\simeq b$. We would like to replace such $a$ and $b$ with $a+b$
so that we have fewer terms to order which is faster and takes less
memory. Fateman investigates the idea of using a hash table on the
terms in the priority queue to collect like terms
\cite{fateHash}. Hash tables do not order their entries, so it is
still necessary to keep a separate priority queue. We say that a
priority queue is \emph{hashed} if it uses a hash table in front of
the priority queue. Fateman reports that a hashed priority queue is
not the best option for monomial multiplication due to the overhead
imposed by the hash table.

Johnson had the idea that instead of keeping track of the terms of
$mg_i$, we could instead have a priority queue containing only the
maximal term of $mg_i$ \cite{origPtrHeap}. Once that term is
extracted, we would then insert the second-most-maximal term of $mg_i$
and so on. This requires annotating values in the priority queue with
information about $m$, about $g_i$ and about which term is the next
one. In this way the priority queue will contain fewer elements which
implies fewer comparisons and a smaller memory footprint. We say that
a priority queue using this idea is \emph{compressed}, since it
compresses information from all of $mg_i$ into a single entry.

When a compressed item in a priority queue is replaced by its
successor term, then we are replacing the maximal value in the
priority queue with a smaller value. Call this pop-push operation
\emph{replace-top}. Many priority queues can do a replace-top
operation faster than a pop followed by a push. Heaps are one
example. The \emph{tournament tree} is especially good at replace-top
operations.  For that reason we investigate using tournament trees in
polynomial division.

We have implemented a heap, a geobucket and a tournament tree for use
in polynomial division as well as hashed and compressed versions of
those data structures. We have made considerable effort to implement
these data structures in an efficient manner --- see
\appRef{app:prio}. We have focused on the general case, so we have not
used packed representations of the monomials.

Table \ref{tab:reducer} compares combinations of these techniques.
When two terms are compared, they might be determined to be equal. In
that case the two terms can be replaced by their sum, though handling
this imposes an overhead. In Table \ref{tab:reducer} the row ``dedup''
indicates whether duplicates are removed in this way. In our
experiment the hashed heap, geobucket and tournament tree have similar
performance and they are faster than the other options. Whether the
dedup and compression options are an advantage depends on the
particular configuration that they are applied to --- see Table
\ref{tab:reducer}. We do not list times for dedup in combination with
hashing, since hashing removes all duplicates.

\subsection{Monomial Ideal Data Structures}
\label{sec:monds}

Monomial ideal computations occur in several places in both signature
and classic \grobner{} basis algorithms. The most apparent example is
in reduction, where it is necessary to determine a basis element whose
lead term divides the term being reduced. This involves deciding the
membership problem on the monomial ideal that is generated by the lead
terms of the basis elements. We call this operation a \emph{divisor
  query}. We investigate data structures for divisor queries and
related operations. See \appRef{app:monds} for more details.

A straight forward divisor query algorithm is to check every monomial
in the data structure for whether it divides the query monomial. We
call this scheme a \emph{monomial list}.

Milowski proposed the \emph{monomial tree} data structure
\cite{monTree}. The monomial tree is a \emph{trie} on the exponent
vectors of the monomial ideal. Milowski shows that this data structure
can be significantly faster than a monomial list in many
cases. Unfortunately the monomial tree degenerates into a
higher-overhead monomial list if all the monomials have distinct
exponents of the first variable.

The toric \grobner{} basis implementation 4ti2 uses an unpublished
binary tree data structure due to Peter Malkin that we will call a
\emph{support tree}. Monomials are stored in the leaves. The leaf that
a monomial $a$ goes into depends on the support of the exponent vector
of $a$. Starting at the root of the tree, go to the right child if
$x_1|a$ and otherwise go left. Do the same thing at the next node for
$x_2$ and so on. A leaf is split into two smaller leaves if it
contains too many monomials. This data structure works well for toric
ideals as about half the exponents are zero. Unfortunately the support
tree degenerates into a higher-overhead monomial list if most of the
monomials have similar support.

We propose the use of \emph{kd-trees}. Kd-trees are used extensively
in computer graphics to keep track of sets of points. The exponent
vector of a monomial is also a point, so kd-trees can also be used as
a data structure for monomial ideals. Both the monomial tree and the
support tree can be described as special cases of kd-trees.

Kd-trees are binary trees. In our kd-tree implementation, the
monomials are in the leaves and each interior node contains a pure
power $x_i^k$. A monomial $a$ goes into the right subtree if $x_i^k|a$
and otherwise it goes into the left subtree. When looking for a
divisor of a monomial $a$, we then do not need to consult the right
subtree if $x_i^k$ does not divide $a$.

\emph{Divmasks} are a widely known technique to speed up divisor
queries. In the most general terms, a divmask involves a function $d$
from monomials to the set of vectors $\set{0,1}^k$ such that if $a|b$
then $d(a)\leq d(b)$. We call such a function a \emph{divmap}. The
idea is then that if $d(a)\nleq d(b)$, then we already know that $a$
does not divide $b$ so we do not have to check it. Furthermore,
checking if one 0-1 vector is dominated by another can be done very
quickly on a computer by letting each entry in the vector $d(a)$
correspond to a bit in a word $w(a)$ of memory. Then $d(a)\leq d(b)$
if and only if the bitwise-and of $w(a)$ with the bitwise negation of
$w(b)$ is zero. In C notation this is \mbox{\tt (a \& $\sim$b) == 0}.

In our implementation the divmaps $d_{x_i^t}$ are parametrized by a
pure power $x_i^t$. Then $d_{x_i^t}(a)=1$ if $x_i^t|a$. We choose the
divmaps based on the monomials in the data structure and periodically
recalculate the divmaps so that they are always appropriate for the
monomials in the datastructure.

The divmask version of our monomial list keeps a divmask for each
monomial. The divmasks eliminate around 98\% of all divisibility
checks for most examples when using the monomial list --- see
\appRef{app:monds}. We also have a divmask version of our kd-tree
where the internal nodes have a divmask of the gcd of the monomials in
that subtree, and the leaves also have a divmask for each
monomial. The subtree rooted at a node does not have to be searched
for a divisor of $a$ if the divmask at that node implies that the gcd
of the monomials in the leaves does not divide $a$.

Table \ref{tab:sbTimes} shows the performance of divmasks, kd-trees
and monomial lists. The baseline is a divmask kd-tree. The kd-tree and
the divmask monomial list are both an improvement on a monomial list,
and the combination of both techniques (baseline) is the fastest in
every case.

\subsection{The S-pair Triangle}
\label{sec:sptri}

There are $\binom{n}{2}$ S-pairs among $n$ basis elements, and for
large $n$ the time spent on ordering those S-pairs that are not
eliminated can be significant. Ordering S-pairs by signature are a
requirement in \galg{}. It it not required but is still a good idea
for the classic Buchberger algorithm since it can be a large speed up
to reduce certain S-pairs first. Other than the time spent on ordering
S-pairs, just storing the S-pairs can also consume a large amount of
memory; especially so for signature algorithms since signature
\grobner{} bases are larger than minimal \grobner{} bases. We present
an S-pair data structure that is fast and uses little memory.

We want to order S-pairs according to some total order $\prec$. We
need a data structure that can give us the minimum S-pair according to
$\prec$, and the data structure needs to support insertion of new
S-pairs every time a new element is added to the basis. So what we
need is a \emph{priority queue} on S-pairs.

A straight forward solution is to use a heap, geobucket or tournament
tree (see Section \ref{sec:prio}) to order the S-pairs. A problem here
is that S-pairs are frequently ordered according to a monomial, such
as in the case of signature algorithms, so it is necessary to store a
monomial for every S-pair in the queue in order to allow fast
comparison by $\prec$. For a basis with 10,000 elements in 50
variables that requires storing up to 5 billion exponents, which at 4
bytes per exponent translates into ~20GB of memory. Many of those
S-pairs are likely going to be eliminated, but the memory overhead is
still ~4GB even if 80\% of the S-pairs are eliminated before putting
them into the queue.

We propose the \emph{S-pair triangle}, which is a priority queue for
S-pairs that only needs to store a single integer per S-pair in the
queue. It is based on the observation that new S-pairs are constructed
in large batches every time a new element is added to the basis. We
sort the new batch of S-pairs according to $\prec$ and maintain a
(small) priority queue on the $\prec$-minimal element from each
batch. If we place the basis elements in a row, and we place each
sorted batch as a column above the corresponding basis element, then
we get the triangle shape that the data structure is named after.

The minimal S-pair in the small priority queue (that has an element
from each column) is also the minimal S-pair over all. To extract the
minimal S-pair $p$, we remove it from the small priority queue and
insert the next-smallest element from the column that $p$ comes from
in the triangle. This is analogous to the compressed priority queues
from Section \ref{sec:prio}, where the columns of the triangle play
the same role as reducers do in the compressed priority queues.

The main attractive property of the S-pair triangle is that we can
throw out the monomials associated to the S-pairs once they are sorted
into columns. We can do so because we only ever need to compare the
$\prec$-minimal S-pair from each column with any other S-pair. So
instead of having to store up to $\binom{n}{2}$ monomials, we only
need to store up to $n$ monomials --- one monomial for each of the $n$
columns. Another memory consumption benefit is that all the S-pairs
$\spair i j$ in column $j$ share the same $j$, so we only need to
store $i$. So we only need memory for one integer per S-pair.

For very large bases even just one integer per S-pair in the queue can
add up to a substantial amount of memory. In our implementation we use
a 16 bit integer for the columns of the first $2^{16}=65,536$ basis
elements, and we then use a 32 bit integer for basis elements beyond
that. This technique halves the memory used on S-pairs in most cases
compared to using a 32 bit integer for all columns.

We have implemented the S-pair triangle with a tournament tree in
front and with a heap in front. We have also implemented an S-pair
priority queue with all the S-pairs in a heap and in a tournament
tree. The baseline is the tournament tree in front of an S-pair
triangle. In Table \ref{tab:sbTimes} we see that only yang1 and mayr42
stress the S-pair queue. For those two examples we see that the
baseline S-pair triangle with a tournament tree in front is the
fastest while the S-pair triangle with a heap in front is a little
slower. The pure heap and tournament tree are much slower and they
make the program consume more than 4 GB of ram on yang1. We conclude
that the S-pair triangle is fast and uses little memory and it works
well with a tournament tree in front.

\section{Experiments}
\label{sec:exp}

We have written an implementation of \galg{} that we use for these
comparisons \cite{onlinePaper}. It has all the improvements that we
present in this paper. Its main current weakness is that it does not
use F4 reduction. We have also started writing a classic Buchberger
implementation for comparison. It benefits from our data structures
and the S-pair elimination component is state of the art but it is
otherwise na\"ive.

We have chosen the examples to show a wide range of behaviors. Some of
the ideals such as yang1 stress the handling of monomials and divisor
queries, while others such as hcyclic8 stress the reduction
procedure. We make the input bases for these examples available online
\cite{onlinePaper}. The example joswig101 uses an elimination order
that eliminates the first 4 variables, breaking ties with grevlex. The
input bases are interreduced.

All benchmarks were run on an Apple MacBookPro intel core i7, at
2.66GHz, with 8 GB of RAM. We use Magma V2.18-1, Maple 14's FGb,
Singular 3-1-3, Macaulay2 version 1.4 (Algorithm LinearAlgebra) and
our classic Buchberger algorithm implementation ``A''. The Macaulay 2
times for inhomogeneous ideals are not given as the F4 implementation
in Macaulay 2 only works for homogenous ideals currently. The FGb
entries with a ``*'' are cases where Maple initially used FGb to do
the calculation, but then terminated FGb and used its own less fast
internal code instead.

We would like to give a definitive discussion about the relative
merits of the various \grobner{} basis algorithms, but we
unfortunately find that the task is currently impossible. The best
\grobner{} basis implementations for the common case of inputs that
cause a lot of time to be spent on polynomial reduction are currently
those in Magma and FGb, and it is not possible to inspect the source
code of either of those systems. As such, there is no way to be
certain about what it is that causes these implementations to work
well. Collecting experimental data without knowing what the experiment
being done actually is is not the highest scientific ideal.

We give an analysis of the differences in Table \ref{tab:examples}
with the caveat that we must necessarily make assumptions that we
cannot verify for the reasons just given.

Our na\"ive Buchberger implementation does very well on yang1 and
mayr42 which we suspect is because the other systems do not use
kd-trees for divisor queries and those two ideals stress the divisor
query infrastructure due to having many variables and many elements in
the final \grobner{} basis. It does poorly on the remaining ideals
because it is not a mature implementation.

Macaulay 2 does well on hcyclic8. That is because of the use of F4, as
the very similar non-F4 Buchberger implementation in M2 (not shown) is
much slower on hcyclic8. FGb and M2 use the same amount of time on
hcyclic8. Our best guess is that Magma is significantly faster here
because its F4 implementation is very good.

Table \ref{tab:sbTimes} shows the times for computing signature bases
on our set of examples.  The baseline algorithm uses a hashed
geobucket for \sigred{}, and a divmask kd-tree for divisor queries and
all optimizations that we have presented.

The term order on $R^m$ for all these experiments is the Schreyer
(induced) order, which is the same as GVW's g2 order.  GVW reports and
our experience confirms that this is often the best order for
computing signature bases.  There are notable exceptions.  For
example, consider the free module order: higher component is greater.
On ties, use the Schreyer order.  Then the signature basis for the
joswig101 example is computed in $6.2$ seconds, with $1242$ basis
elements, much faster than all the \grobner{} times reported for that
example in Table \ref{tab:examples}.  However, the result is still
much larger than the reduced \grobner{} basis.

The \galg{} implementation is the fastest on jason210 and performs
reasonably on the other ideals except for yang1. The reason for the
slow performance on yang1 is that the signature \grobner{} basis is
much larger than the minimal \grobner{} basis for yang1. For the other
ideals we suspect that FGb and Magma are faster because they use F4
reduction, so we suspect that the comparison is not useful for
determining if \galg{} is faster than F5.

Recall that \galg{} computes the initial module of the syzygy module
of the original basis. GVW explain how to compute the \grobner{} basis
of the syzygy module from this initial syzygy module. Since \galg{}
often computes the initial syzygy module in about the same time as it
takes to compute a \grobner{} basis for most of these examples,
\galg{} should be the best algorithm for computing \grobner{}
bases of syzygy modules.

     \begin{table*}
       \centering
       \begin{tabular}{l|rrrllr|rrrrr|r}
         Example&$p$&nvars&neqns&homog?& order & nGB & magma & FGb & Sing & M2 & A & SB \\
         \hline
         joswig101 & 101 & 5 & 5 & no & elim(4) & 5 & 56.3 & n/a & 122.6 & n/a & * & 93.0\\
         jason210 & 32003 & 8 & 3  & yes & grevlex & 900 & 6.0& 57.5 & 2.3 &         12.1  & 4.6 & 1.4\\
         katsura10 & 101 & 10 & 10 & no & grevlex & 272 & .5 & 1.9 & 7.6 &  n/a  & 22.9 & 2.8 \\
         katsura11 & 101 & 11 & 11 & no & grevlex & 537 & 3.5 & 13.2 & 63.1 &
         n/a & 253.0 & 18.4 \\
         hcyclic8 & 101 & 9 & 8 & yes & grevlex & 1182 & 3.5 & 12.5 & 43.0 &  12.5 & 162.0 & 111.5\\
         yang1 & 101 & 48 & 66 & yes & grevlex & 4761 & 29.0 & * & 85.3 & 64.1   & 5.6 & 1333.0\\
         mayr42 & 101 & 51 & 44  & yes & grevlex & 8534 & 54.2 & 347.0 & 218.3   & 89.6 & 26.9 & 273.0\\
       \end{tabular}
       \caption{Input data and time in seconds for several implementations.}
       \label{tab:examples}
     \end{table*}

\setlength{\tabcolsep}{4pt}

     \begin{table*}
     \centering
     \begin{tabular}{l|rrrrrrr}&joswig101&jason210&katsura10&katsura11&hcyclic8&yang1&mayr42 \\
     \hline
     \#spairs: &1,209,790&987,715&37,950&148,240&14,680,071&1,998,099,720&523,633,341 \\
     \hline
     elim via non-regular criterion&655&1,191&206&698&2,821&111,120&362,703 \\
     elim via base divisor criterion&686,300&346,714&14,864&62,634&6,711,383&1,415,552,384&364,970,054 \\
     elim via signature criterion&281,682&573,998&10,998&46,170&7,154,919&409,866,276&149,612,312 \\
     \hline
     \#spairs queued&241,153&65,812&11,882&38,738&810,948&172,569,940&8,688,272 \\
     \hline
     elim via duplicate signature&219,554&55,475&9,283&32,834&713,321&116,423,105&4,337,133 \\
     elim via signature criterion(late)&16,720&5,235&2,110&4,967&75,987&45,461,188&4,036,194 \\
     elim via Koszul criterion&11&1&71&88&30&14,165&376 \\
     elim via rel. prime criterion&0&2&71&148&7&31,762&5,507 \\
     elim via singular criterion(late)&3,101&3,338&0&0&15,430&10,490,908&111,466 \\
     \hline
     \#spairs which need reduction&1,767&1,761&347&701&6,173&148,812&197,596 \\
     reduce to SB elements&1,551&1,403&266&534&5,411&63,150&32,318 \\
     reduce to new syzygy signatures&216&358&81&167&762&85,662&165,278 \\

     \end{tabular}
     \caption{Number of S-pairs eliminated by the various criteria in \galgex{}.}
     \label{tab:sigSPair}
     \end{table*}

\setlength{\tabcolsep}{6pt}

\begin{table*}
\centering
\begin{tabular}{l|rrrrrrr}&joswig101&jason210&katsura10&katsura11&hcyclic8&yang1&mayr42 \\
\hline
\#SB&1,556&1,406&276&545&5,419&63,216&32,362 \\
\#monomials&760,690&519,315&100,626&387,769&3,281,515&1,224,044&64,724 \\
\hline
baseline&93&1&3&18&112&1333&273 \\
no fast ratio&134&1&2&20&120&2753&565 \\
no base divisors&90&1&2&19&112&2022&478 \\
early koszul&93&2&2&18&115&1341&337 \\
\hline
divmask monomial list&84&2&2&19&139&3917&835 \\
monomial list&179&9&4&37&1270&$> 8$ hours&$> 30$ min \\
kd-tree&100&2&2&19&119&2113&419 \\
\hline
spair-tourTree&93&1&2&19&114&$> 4$ GB&320 \\
spair-heap&118&2&3&24&147&$ > 4$ GB&379 \\
spair-heap-triangle&92&1&2&18&112&1488&277 \\

\end{tabular}
\caption{Time in seconds for variants of \galgex{}.}
\label{tab:sbTimes}
\end{table*}

\setlength{\tabcolsep}{3.5pt}

     \begin{table*}
     \centering
       \small
       \begin{tabular}{l|rrrrrr|rrrrrr|rrrrrr}
       Reducer & \multicolumn{6}{|c|}{Tour tree} & \multicolumn{6}{c}{Heap} &
         \multicolumn{6}{c}{Geobucket} \\
       \hline
       Hashed & x & x & . & . & . & . &      x & x & . & . & . & . &      x & x & . & . & . & . \\
       Dedup  & . & . & x & x & . & . &      . & . & x & x & . & . &      . & . & x & x & . & . \\
       Compressed & x & . & x & . & x & . &      x & . & x & . & x & . &      x & . & x & . & x & . \\
       \hline
     joswig101 & 101 & 95 & 247 & 640 & 226 & 614 & 104 & 93 & 173 & 655 & 241 & 609 & 96 & 93 & 124 & 139 & 206 & 439\\
     jason210 & 1 & 2 & 2 & 4 & 2 & 4 & 1 & 2 & 2 & 4 & 2 & 3 & 1 & 2 & 2 & 2 & 2 & 3\\
     katsura10 & 3 & 2 & 12 & 37 & 10 & 33 & 3 & 2 & 6 & 35 & 11 & 31 & 3 & 2 & 4 & 6 & 11 & 21\\
     katsura11 & 21 & 19 & 106 & 373 & 91 & 348 & 21 & 19 & 57 & 404 & 105 & 357 & 21 & 19 & 39 & 45 & 97 & 245\\
     hcyclic8 & 121 & 113 & 538 & 2062 & 470 & 1848 & 123 & 117 & 296 & 2209 & 517 & 1968 & 125 & 116 & 230 & 258 & 505 & 1275\\
     yang1 & 1330 & 1330 & 1329 & 1335 & 1330 & 1336 & 1330 & 1332 & 1329 & 1334 & 1330 & 1334 & 1337 & 1339 & 1338 & 1335 & 1338 & 1335\\
     mayr42 & 274 & 273 & 273 & 274 & 273 & 274 & 272 & 279 & 273 & 274 & 273 & 274 & 282 & 273 & 272 & 277 & 273 & 277\\

       \end{tabular}
     \caption{Time in seconds using different reducer data structures in \galgex{}.}
     \label{tab:reducer}
     \end{table*}

\hideAppendix{
\appendix
\section{The \galg{} Algorithm}
\label{app:alg}

These appendices contain material that we had to cut to fit the page
within the ISSAC page limit of 8 pages. In particular the appendices
contain proofs of all the theorems from the paper. The theorem numbers
are not consecutive because we repeat the theorems that appear in the
main paper and they retain their original number.

\subsection{Setup}

Our notation and terminology differs from both that of Gao, Volny and
Wang and of Arri and Perry. Part of the difference is that we attempt
to make the language and description of \galgex{} as close as possible
to that for the classic Buchberger algorithm.

All notation for signature algorithms have used a pair $(f,s)$ where
$s$ is the signature of $f$. Gao, Volny and Wang let $f\in R^n$ rather
than $f\in R$. We consider it an advantage not to have three symbols
$p=(f,s)$ tied up for every pair, and we also avoid any ambiguity
about whether the word ``pair'' refers to an S-pair or a basis
element. We have not used the natural concept that $\hd u=\sig u$,
since then we can talk about both the signature $\sig u$ and lead term
$\hdp u$ of $u$ without ambiguity. The use of a bar to signify the
mapping $u\mapsto\proj u$ is also not standard, but it is convenient
since our other notation requires heavy use of the mapping.

An implementation of the algorithm does not need to maintain a full
representation $u\in R^n$ of each polynomial $\proj u$ just because
the mathematical arguments concern an element of $R^n$. An
implementation only needs to store $\sig u$ and $\proj u$.

If $a\simeq b$ then we say that $a$ and $b$ are \emph{like}, as in the
phrase ``collecting like terms''. In section \ref{sec:algsetup} we
define the signature of 0. This is never relevant for \galgex{} as we
never add two elements with like signatures, so there is no occasion
for the coefficient of the signature to become zero. It might seem
that we should simply say that the zero syzygy does not have a
signature and leave it at that. However, even non-zero syzygies can
have a zero signature. This happens when a non-zero syzygy $v\in R^n$
maps to $\phi(v)=0\in R^m$. Even though that is never relevant for
\galgex{}, defining $\phi(v)$ in all cases allows us to avoid
addressing the issue in the paper.

\subsection{Division With Signatures}

The definition of \sigdiv{} is intended to be as close as possible to
the definition of classic polynomial division. We had originally
described classic polynomial division in the paper, and then
introduced \sigdiv{} after that with as few changes as possible to
underscore the analogy. Our original classic polynomial division
description is here, as well as pseudo code for classic polynomial
division, \sigdiv{} and regular division.

We had originally referred to \sigred{} as \emph{signature reduction},
but we adopted Arri and Perry's name for it because it saved space and
we had tremendous trouble fitting the paper into the limit of 8
pages. We debated referring to \sigred{} as just \emph{reduction} with
no prefix, but that risked confusing \sigred{} with one of regular
reduction, classic polynomial reduction and a singular reduction step.

\subsubsection*{Polynomial division in $R$}
\label{app:classicDiv}

The usual notion of polynomial division is to divide a polynomial
$f\in R$ by a finite set of polynomials $\extendBasis n\subseteq
R$. This yields a quotient $q\in R^n$ and a remainder $r\in R$ such
that
\begin{enumerate}
\item $f=\proj q + r$,
\item $\hd f\geq\hd{q_ig_i}\text{ for }i=1,\ldots,n$,
\item no term of $r$ is divisible by any $\hd{g_i}$ for $g_i\in\extendBasis n$.
\end{enumerate}
If $f$ is zero then so are $q$ and $r$. We \emph{divide} to get the
quotient $q$ and we \emph{reduce} to get the remainder $r$. The
process is the same so the distinction between division and reduction
is just what part of the result $(q,r)$ that we most care about. We
say that $f$ is \emph{reduced} if $q=0$ and otherwise $f$ is
\emph{reducible}.

If $r=0$ then call $q$ a \emph{representation} of $f$ and then $f$ has
a representation. Since polynomial division by something that is not a
\grobner{} basis does not have a unique outcome, $f$ might have a
representation even if the polynomial division algorithm does not
arrive at a zero remainder.

To compute $(q,r)$ we perform reduction steps. If $t$ is a term of $f$
then we can reduce that term by $g_i\in\extendBasis n$ when $\hd{g_i}$
divides $t$. Let $a$ be a monomial such that $\hd{ag_i}=t$. Then we
perform a reduction step by subtracting $ag_i$ from $f$ and adding
$a\mbasis i$ to $q$. We continue this process to complete the
reduction.

There is a distinction between reduction and \emph{top reduction}. In
top reduction the reduction is halted as soon as the initial term of
$f$ cannot be reduced. We say that $f$ is \emph{top reducible} if its
initial term can be reduced and otherwise it is \emph{top reduced}.

The following pseudo code implements polynomial division in $R$. It
returns a pair $(q,r)$ where $q\in R^n$ is the quotient and $r\in R$
is the remainder.
\ \\
\ \\
{\bf Reduce($f\in R$)}
\begin{algorithmic}
\STATE $q\gets 0\in R^n$\quad \COMMENT{$q$ is the quotient}
\STATE $r\gets 0\in R$\quad \COMMENT{$r$ is the sum of those terms of $f$ that we have not been able to reduce}
\WHILE{$f\neq r$}
  \STATE $t\gets\hd{f-r}$\quad\COMMENT{$t$ is the maximal term of $f$ that we have not yet processed}
  \STATE $d\gets 0$\quad\COMMENT{$d$ will store any divisor that we may find}
  \FOR[look for a divisor]{$i=1,\ldots,n$}
    \IF{$\hd{g_i}|t$}
      \STATE $d\gets\frac{t}{\hd{g_i}}\mbasis i$\quad\COMMENT{for efficiency you would stop the for loop here}
    \ENDIF
  \ENDFOR
  \IF{$d\neq 0$}
    \STATE $f\gets f-\proj d$\quad\COMMENT{reduce by $d$}
    \STATE $q\gets q+d$\quad\COMMENT{record that we reduced by $d$}
  \ELSE
    \STATE $r\gets r+t$\quad\COMMENT{record that the term $t$ could not be reduced}
  \ENDIF
\ENDWHILE
\RETURN $(q,r)$\quad\COMMENT{$q\in R^n$ is the quotient and $r\in R$ is the remainder}
\end{algorithmic}

\subsubsection{\sigred{} in $R^n$}

The following pseudo code implements \sigred{} in $R^n$. It returns a
pair $(q,r)$ where $q\in R^n$ is the quotient and $r\in R^n$ is the
remainder.
\ \\
\ \\
{\bf SReduce($u\in R^n$)}
\begin{algorithmic}
\STATE $T\gets\sig u$\quad\COMMENT{The signature of $u$ can change, so we have to record it}
\STATE $q\gets 0\in R^n$\quad\COMMENT{$q$ is the quotient}
\STATE $r\gets 0\in R$\quad\COMMENT{$r$ is the sum of those terms of $\proj u$ that we have not been able to reduce}
\WHILE{$\proj u\neq r$}
  \STATE $t\gets\hd{\proj u-r}$\quad\COMMENT{$t$ is the maximal term of $\proj u$ that we have not yet processed}
  \STATE $d\gets 0$\quad\COMMENT{$d$ will store any divisor that we may find}
  \FOR[look for a divisor]{$i=1,\ldots,n$}
    \IF{$\hd{g_i}|t\text{ and }\sig{\frac{t}{\hd{g_i}}\mbasis i}\leq\sig u$}
      \STATE $d\gets\frac{t}{\hd{g_i}}\mbasis i$\quad\COMMENT{for efficiency you would stop the for loop here}
    \ENDIF
  \ENDFOR
  \IF{$d\neq 0$}
    \STATE $u\gets u-d$\quad\COMMENT{reduce by $d$}
    \STATE $q\gets q+d$\quad\COMMENT{record that we reduced by $d$}
  \ELSE
    \STATE $r\gets r+t$\quad\COMMENT{record that the term $t$ could not be reduced}
  \ENDIF
\ENDWHILE
\RETURN $(q,u)$\quad\COMMENT{$q\in R^n$ is the quotient and $u\in R^n$ is now the remainder as $\proj u=r$}
\end{algorithmic}

\subsubsection{Regular Division in $R^n$}

The result of regular division of $u\in R^n$ by $\extendBasis n$ is a
quotient $q\in R^n$ and a remainder $r\in R^n$ such that
\begin{enumerate}
\item $u=q+r$,
\item $\hdp u\geq\hd{q_ig_i}\text{ for }i=1,\ldots,n$,
\item if $\sig u>\sig{a\mbasis i}$ for a monomial $a$ then $\hdp{a\mbasis
  i}$ does not equal any term of $r$,
\item $\sig u>\sig q$.
\end{enumerate}
These conditions are identical to the ones for \sigred{} except that
``$\geq$'' has been replaced with ``$>$'' for the last condition. In
the same way, the condition for $a\mbasis i$ to regular reduce a term
$t$ of $\proj u$ becomes
\[
\hdp{a\mbasis i}=t
\quad\text{ and }\quad
\sig{a\mbasis i}<\sig u.
\]
Note that $\sig{a\mbasis i}<\sig u$ is equivalent to $\sig
u=\sig{u-a\mbasis i}$. To see this, recall that the signature includes
a coefficient. So a regular reduction step happens when it can be
carried out without changing the signature.

The pseudo code for \sigred{} can be modified to carry out regular
reduction by replacing the line

\begin{algorithmic}
\STATE\quad\quad{\bf{}if }$\hd{g_i}|t\text{ and }\sig{\frac{t}{\hd{g_i}}\mbasis i}\leq\sig u${\bf{} then}
\end{algorithmic}

with

\begin{algorithmic}
\STATE\quad\quad{\bf{}if }$\hd{g_i}|t\text{ and }\sig{\frac{t}{\hd{g_i}}\mbasis i}<\sig u${\bf{} then}
\end{algorithmic}

We define \emph{regular division}, \emph{regular reduction},
\emph{regular reduced}, \emph{regular reducible}, \emph{regular top
  reduced} and \emph{regular top reducible} analogously to how those
terms are defined for \sigred{}.

\subsection{S-pairs}

We give the proofs that are missing from Section \ref{sec:spairs}. The
arguments are essentially the same as the ones given by Gao, Volny and
Wang, though stated in a different way.

\repthm{theorem}{thm:spairs}{
Let $T$ be a term of $R^m$. Assume for all S-pairs $p$ with $\sig
p\leq T$ that if $p\p$ is the result of regular reducing $p$, then
$p\p$ is singular top reducible or a syzygy. Then all elements $u\in
R^n$ with $\sig u\leq T$ \sigrede{} to zero.
}
\begin{proof}
Suppose to get a contradiction that there is a $u\in R^n$ with $\sig
u\leq T$ such that $u$ does not \sigrede{} to zero. Assume without
loss of generality that $\sig u$ is $\leq$-minimal such that $u$ does
not reduce to zero. We may also assume that $u$ is top reduced.

By Lemma \ref{lem:p1} there is an S-pair $p$ whose signature divides
$\sig u$. Also, $ap\p$ is regular top reduced where $p\p$ is the
result of regular reducing $p$ and $a$ is the monomial such that
$\sig{ap}=\sig u$.

Now $\sig{ap\p}=\sig u$ and both $ap\p$ and $u$ are regular top
reduced, so by Lemma $\ref{lem:p3}$ we get that $\hdp{ap\p}=\hdp
u$. Then anything that top reduces $ap\p$ will also top reduce $u$. We
know that $ap\p$ is top reducible since $p\p$ is top reducible by
assumption. Thus $u$ is top reducible which is a contradiction.
\end{proof}

\repthm{lemma}{lem:p3}{
Let $L\in R^m$ be a term such that all $v\in R^n$ with $\sig v<L$
\sigrede{} to zero. Let $a,b\in R^n$ such that $\sig a=\sig b\leq
L$. Then $\hdp a=\hdp b$ if $a$ and $b$ are regular top reduced. Also,
$\proj a=\proj b$ if $a$ and $b$ are regular reduced.
}
\begin{proof}
It suffices to prove that $\proj a=\proj b$ if $a$ and $b$ are regular
reduced as the other statement is a corollary of that.

Suppose to get a contradiction that $\proj{a-b}\neq 0$. As
$\sig{a}=\sig{b}=L$ we get that $\sig{a-b}<L$ so $a-b$ reduces to
zero. Then in particular $a-b$ is top reducible. Swap $a$ and $b$ if
necessary to ensure that $\hdp{a-b}$ has the same monic part as a term
in $a$. Then that term of $a$ is \emph{regular} reducible since
$\sig{a-b}<\sig a$. This contradicts the assumption that $a$ is
regular reduced.
\end{proof}

\begin{lemma}
\label{lem:p1}
Let $u\in R^n$ be top reduced and have non-zero signature. Assume that
all $v\in R^n$ with $\sig v<\sig u$ reduce to zero. Then there exists
an S-pair $p$ whose signature divides the signature of $u$. Also,
$cp\p$ is regular top reduced where $p\p$ is the result of regular
reducing $p$ and $c$ is the monomial such that $\sig{cp}=\sig u$.
\end{lemma}
\begin{proof}
We construct an S-pair $p$ whose signature divides $\sig u$. Then the
rest follows from Lemma \ref{lem:p2}

\proofPart{\text{Consider initial terms}} If $\alpha,\beta\in R^n$ are
both top reducible and $\sig\alpha\leq\sig{\alpha+\beta}$ and
$\sig{\beta}\leq\sig{\alpha+\beta}$ then $\alpha+\beta$ is top
reducible or $\hdp{\alpha}+\hdp{\beta}=0$. We will use this argument
with $\alpha\defeq u-\sig u$ and $\beta\defeq\sig u$.

Let $L\defeq\sig u$. Then $u-L$ has smaller signature than $u$ does so
it reduces to zero and in particular it is top reducible. Also $L$ is
top reducible since it top reduces itself. Yet the sum $(u-L)+L=u$ is
not top reducible so
\[
\hdp{u-L}+\hdp{L}=0.
\]

\proofPart{\text{Construct $p$}}
Let $a\mbasis i\defeq L=\sig u$. As $u-L$ reduces to zero there is a
reducer $b\mbasis j$ such that
\[
\hdp{b\mbasis j}=\hdp{u-L},
\quad\quad
\sig{b\mbasis j}\leq\sig{u-L}.
\]
Let $p\defeq\spair i j$. Then $p=\frac{a}{c}\mbasis
i+\frac{b}{c}\mbasis j$ where $c\defeq\gcd(a,b)$ and $\sig{cp}=L$ since
\[
\hdp{a\mbasis i}=\hdp L=-\hdp{u-L}=-\hdp{b\mbasis j}
\]
and
\[
\sig{a\mbasis i}>\sig{u-L}\geq\sig{b\mbasis j}.
\]
So $p$ is an S-pair whose signature divides $L=\sig u$.
\end{proof}

\begin{lemma}
\label{lem:p2}
Let $L\in R^m$ be a term such that all $v\in R^n$ with $\sig v<L$
reduce to zero. Let $p$ be an S-pair whose signature divides $L$. Then
there exists an S-pair $q$ whose signature also divides $L$ and with
the following additional property. Let $b$ be the monomial such that
$\sig{bq}=L$ and let $q\p$ be the result of regular reducing $q$. Then
$bq\p$ is not regular top reducible.
\end{lemma}
\begin{proof}
Pick a monomial $a$ such that $\sig{ap}=L$ and let $p\p$ be the result
of regular reducing $p$. We can assume that $ap\p$ is regular top
reducible as otherwise we are done. This implies that $a>1$ so $\sig
p<L$ so $p$ reduces to zero.

We are now going to construct an S-pair $q$ such that $\sig{bq}=L$ and
$\hdp{ap}>\hdp{bq}$. If $bq\p$ is not regular top reducible then we
are done. Otherwise we can do the same thing again to get yet a third
S-pair with the same properties and so on. This process must terminate
as the initial terms of the S-pair multiples $ap$, $bq$, $\ldots$ decrease
strictly at each step and there are only finitely many S-pairs
(alternatively, $\leq$ is a well order).

\proofPart{\text{Construct }c\mbasis i} As $\sig{p\p}=\sig p<L$ we
know that $p\p$ reduces to zero. Yet $\proj{p\p}$ is non-zero, so
there must be a singular reducer $c\mbasis i$ such that
\[
\hdp{c\mbasis i}=\hdp{p\p},
\quad\quad
\sig{c\mbasis i}\simeq\sig{p\p}.
\]
Recall that $\simeq$ is equality up to a non-zero element of the
ground field.

\proofPart{\text{Construct }d\mbasis j} As $ap\p$ is regular top reducible
there is a regular top reducer $d\mbasis j$ such that
\[
\hdp{d\mbasis j}=\hdp{ap\p},
\quad\quad
\sig{d\mbasis j}<\sig{ap\p}.
\]

\proofPart{\text{Construct }q} Let $q\defeq\spair i j$. Then
$q=\frac{ac}{b}\mbasis i-\frac{d}{b}\mbasis j$ where $b\defeq\gcd(ac,d)$ and
$\sig{bq}\simeq\sig{ap\p}=L$ since
\[
\hdp{ac\mbasis i}=\hdp{ap\p}=\hdp{d\mbasis j},
\quad\quad
\sig{ac\mbasis i}\simeq\sig{ap\p}>\sig{d\mbasis j}.
\]
Then $\hdp{d\mbasis j}>\hdp{bq}$ as the initial term cancels in the
subtraction $\proj{ac\mbasis i-d\mbasis j}=\proj{qb}$, so as claimed we
see that
\[
\hdp{ap}\geq\hdp{ap\p}=\hdp{d\mbasis j}>\hdp{bq}.
\]
\end{proof}

\repthm{theorem}{thm:syzygies}{
Let $\extendBasis n$ be a signature \grobner{} basis and let $u\in
R^m$ be a syzygy. Then there is an S-pair $p$ that regular reduces to
a syzygy $p\p$ such that $\sig{p\p}$ divides $\sig u$.
}
\begin{proof}
We see that $u$ is reduced since it is a syzygy. Then by Lemma
\ref{lem:p1} there exists an S-pair $p$ whose signature divides the
signature of $u$. Also, $ap\p$ is regular reduced where $p\p$ is the
result of regular reducing $p$ and $a$ is the monomial such that
$\sig{ap}=\sig u$. Then Lemma \ref{lem:p3} implies that
$\proj{ap\p}=\proj u=0$. So $\proj{p\p}=0$ and we are done.
\end{proof}

\subsection{Termination}
\label{app:term}

Huang proves that a GVW-like algorithm terminates
\cite{huangTermination}, and Gao, Volny and Wang refer to Huang for
termination. Huang gives a counterexample to show that \galg{} does
not always terminate when the term order on the module and the term
order on the ring are not compatible in the sense that $a<b\biimp
a\mbasis i<b\mbasis i$.

Eder and Perry \cite{ederPerryF5Like} prove that an F5-like algorithm
terminates using an incremental term order on the module (``position
over term''). We give Perry and Eder's proof here. Their proof
requires no changes to apply to \galgex{} other than to be adapted to
our notation.

We do not need to mention the issue of compatibility of the term
orders in Theorem \ref{thm:term} since that assumption already appears
in Section \ref{sec:algsetup}.

\begin{theorem}
\label{thm:term}
Suppose that $\mmor{\mbasis i}$ is not top reducible by $\extendBasis
{i-1}$ and that all $v\in R^m$ with $\sig v<\sig{\mbasis i}$ do
\sigrede{} to zero with respect to $\extendBasis{i-1}$ for each $i\geq
m$. Then the sequence of $g_i$'s is finite.
\end{theorem}
\begin{proof}
Let $R\p$ be a polynomial ring containing all the variables
$x_1,\ldots,x_k$ of $R$. Also, let $R\p$ contain variables $y_{ij}$
for $i=1,\ldots,m$ and $j=1,\ldots,k$. Define the function
$f:R\times\set{\text{terms of }R^m}\rightarrow R\p$ by
\[
f(g,sx^v\mbasis i)\defeq\hd{g}y_i^vy_{i1}
\]
where $s$ is a non-zero element of the ground field. Then
$f(g,T)|f(g\p,T\p)$ if and only if both $\hd g|\hd{g\p}$ and $T|T\p$.

Consider the sequence of monomial ideals $I_m$, $I_{m+1}$, $\ldots$
defined by
\[
I_n\defeq \ideal{f(\proj{\mbasis 1},\sig{\mbasis 1}),\ldots,f(\proj{\mbasis n},\sig{\mbasis n})}\subseteq
R\p.
\]
If $f(\proj{\mbasis i},\sig{\mbasis i})|f(\proj{\mbasis
  j},\sig{\mbasis i})$ for $i<j$ then $\mmor{\mbasis j}$ is top
reducible by Lemma \ref{lem:term}. We have assumed that none of the
$\mmor{\mbasis j}$ are top reducible, so the sequence of monomial
ideals $I_n$ is strictly increasing and therefore finite.
\end{proof}

\begin{lemma}
\label{lem:term}
Let $u\in R^n$ such that all $v\in R^n$ with $\sig v<\sig u$
\sigrede{} to zero. If $\hdp{\mbasis i}|\hdp u$ and $\sig{\mbasis
  i}|\sig u$ then $u$ is top reducible.
\end{lemma}
\begin{proof}
Let $a$ and $b$ be monomials such that
\[
\hdp{a\mbasis i}=\hdp u
\quad\text{ and}\quad
\sig{b\mbasis i}=\sig u.
\]
If $a\leq b$ then $a\mbasis i$ top reduces $u$ and we are
done. Otherwise $a>b$ so that
\[
\hdp{u-b\mbasis i}=\hdp u
\quad\text{ and }\quad
\sig{u-b\mbasis i}<\sig u.
\]
Then $u-b\mbasis i$ is top reducible and whatever top reduces it will
top reduce $u$.
\end{proof}

\subsection{Pseudo Code}

We have to reiterate that the pseudo code given in Section
\ref{sec:pseudo} is a simplest possible version of the algorithm. Any
reasonable implementation would include at least the S-pair
elimination criteria introduced in Section \ref{sec:sigImp} as well as
many other improvements that have been developed for \grobner{} basis
computation. The pseudo code is intended as a way to succintly state
the essence of \galgex{} without getting bogged down in the
complexities of an efficient implementation.

\section{Signature Improvements}

We add notes to the material in Section \ref{sec:sigImp} and give the
proofs that we did not have space to include in the main paper.

\subsection{S-pair Elimination}
\label{app:spairElim}

Something that we did not have space to dwell on in the main paper is
the widely understood distinction between eliminating an S-pair
\emph{early} and eliminating it \emph{late}. There are two points in
an S-pair's life time where it is natural to try to eliminate it. The
first opportunity is right as it gets created, and the second
opportunity is right before it would otherwise cause a reduction to be
carried out.

Everything else being equal, it is better to eliminate an S-pair early
rather than late. The reason for that is that if an S-pair is
eliminated late, then it had to be stored for a time, which increases
memory consumption, and it had to be compared to other S-pairs to
determine which S-pair has the minimal signature, which takes
time. Sometimes an S-pair can only be eliminated late, for example
when the syzygy signature that eliminates the S-pair is only
discovered after the S-pair is constructed.

It would also be possible to eliminate an S-pair sometime between
early and late. However, as there are many S-pairs, trying to
eliminate them all every time new information comes in would take a
lot of time. The S-pair triangle (see Section \ref{sec:sptri})
minimizes the overhead of storing many S-pairs, so we are not too
concerned about eliminating S-pairs early, even though we do want to
do so when possible.

As we state at the end of Section \ref{sec:spairElim}, Arri and Perry
remark \cite[Remark 20]{apSyzygyF5} that it is possible to apply the
singular criterion early. We have implemented this technique and we
report times for using it in Table \ref{tab:sbTimes}. From that table
we see that applying the singular criterion early causes our program
to run slower, for example for mayr42 the time increases from 273s to
337s. This is due to the extra time it takes to check the singular
criterion on all of the S-pairs that are not eliminated by the other
early criteria. We have used a kd-tree to check the singular criterion
--- otherwise it would have been much slower. Applying the singular
criterion early does decrease the number of queued S-pairs by 53\% for
yang1, which is an advantage since it decreases the amount of memory
used to store pending S-pairs.

We postpone the duplicate criterion and the relatively prime criterion
to maximize the number of S-pairs that are eliminated. Let $p$ and $q$
be two S-pairs with the same signature and suppose that the relatively
prime criterion can eliminate $p$. We can eliminate $q$ due to the
duplicate criterion and then eliminate $p$ due to the relatively prime
criterion. The order there is important. If we first eliminate $p$ due
to the relatively prime criterion, then we will not have any way to
eliminate $q$. We postpone the two criteria so that we can check if
any S-pair in a given signature is relatively prime. In contrast, we
do not postpone the signature criterion because if it applies to one
S-pair in a signature, then it applies to all of them so there is no
reason to delay.

Note that criteria that get applied early look more impressive in
Table \ref{tab:sigSPair} because they get checked for many more
S-pairs. For example if an S-pair can be eliminated both by the
signature criterion and the Koszul criterion, then that will count
only as a hit for the signature criterion since that S-pair is then
eliminated so that the Koszul criterion never gets the opportunity to
eliminate it.

\repthm{corollary}{col:spairSyzygy}{
Let $u\in R^n$ such that all $v\in R^n$ with $\sig v<\sig u$ reduce to
zero. Suppose there exists a syzygy $h\in R^n$ whose signature divides
the signature of $u$. Then $u$ regular reduces to zero.
}
\begin{proof}
Let $u\p$ be the result of regular reducing $u$. Let $a$ be the
monomial such that $\sig{ah}=\sig u$. Now $ah$ and $u\p$ have the same
signature and they are both regular reduced so Lemma \ref{lem:p3}
implies that $\proj{u\p}=\proj{ah}=0$.
\end{proof}

\repthm{corollary}{cor:p4}{
Let $p$ be an S-pair and let $p\p$ be the result of regular reducing
$p$. Let $M$ be the finite set
\[
M\defeq\setBuilder{a\mbasis i}{a\text{ is a monomial and }\sig{a\mbasis i}=\sig p}.
\]
Then all elements of $M$ regular reduce to $p\p$. Also, $p\p$ is
singular top reducible if and only if some element of $M$ is regular
top reduced.
}
\begin{proof}
All elements in $M$ regular reduce to $p\p$ by Lemma
\ref{lem:p3}. Suppose that $a\mbasis i\in M$ is regular top
reduced. Then $\hdp{a\mbasis i}=\hdp{p\p}$ by Lemma \ref{lem:p3} so
$a\mbasis i$ top reduces $p\p$. Suppose instead that $p\p$ is top
reducible. Then there is a singular top reducer $a\mbasis i$ such that
$\hdp{a\mbasis i}=\hdp{p\p}$ and $\sig{a\mbasis
  i}\simeq\sig{p\p}=\sig{p}$.  Then $a\mbasis i$ is regular top
reduced since $p\p$ is. Also, there is an element $s$ of the ground
field such that $sa\mbasis i\in M$.
\end{proof}

\subsection{Base Divisors}

The base divisor criterion from Section \ref{sec:basediv} is mainly
useful when there are a very large amount of basis elements, as that
is when handling S-pairs can take a lot of time. The technique also
requires storing a triangle of $\binom{k}{2}$ bits where $k$ is the
size of the basis. That can take a significant amount of memory when
$k$ is very large. So an overhead in memory is only imposed when there
is at the same time a significant advantage in time.

Running out of memory is worse than taking a little longer, since the
computation cannot proceed if there is not enough memory. However, if
the computer runs out of memory when using the base divisor technique,
then it is possible to simply drop the triangle of bits and stop using
the base divisor technique from that point onward. So in this way we
can view the base divisor technique as a way to use spare memory to
speed up the computation, but that memory can be freed if it is
needed.

We give the proofs from Section \ref{sec:basediv} that we did not have
space to include in the main paper.

\repthm{theorem}{thm:highRatio}{
Let $\alpha,\beta,\gamma\in R^n$ such that $\hdp\alpha|\hdp\beta$ and
$\frac{\sig\gamma}{\hdp\gamma}> \frac{\sig\alpha}{\hdp\alpha},
\frac{\sig\beta}{\hdp\beta}$. Then
$\sig{\spair\alpha\gamma}|\sig{\spair\beta\gamma}$.
}
\begin{proof}
To ease notation, let $x^a\defeq\hdp\alpha$, $x^b\defeq\hdp\beta$ and
$x^c\defeq\hdp\gamma$. The assumptions about sig-lead ratios imply
that
\[
\sig{\spair\alpha\gamma}=
\frac{\hdp\alpha}{\gcd(\hdp\alpha,\hdp\gamma)}\sig\gamma
\]
and that
\[
\sig{\spair\alpha\gamma}=
\frac{\hdp\beta}{\gcd(\hdp\beta,\hdp\gamma)}\sig\gamma.
\]
So in vector notation we need to prove for each $i$ that
\[
\min(b_i,c_i)-\min(a_i,c_i)\leq b-a.
\]

\proofPart{\text{The case }a_i,b_i\geq c_i} Equivalent to $a_i\leq b_i$.

\proofPart{\text{The case }a_i> c_i > b_i} Does not happen as $a_i\leq
b_i$.

\proofPart{\text{The case }a_i\leq c_i\leq b_i} Equivalent to $c_i\leq
b_i$.

\proofPart{\text{The case }a_i,b_i\leq c_i} Equivalent to $b_i-a_i\leq
b_i-a_i$.
\end{proof}

\repthm{theorem}{thm:lowRatio}{
Let $\alpha,\beta,\gamma\in R^n$ such that $\sig\alpha|\sig\beta$ and
$ \frac{\sig\gamma}{\hdp\gamma}< \frac{\sig\alpha}{\hdp\alpha},
\frac{\sig\beta}{\hdp\beta}$. Let
$x^p\defeq\frac{\hdp\alpha\sig\beta}{\sig\alpha}$,
$x^a\defeq\hdp\alpha$ and $x^b\defeq\hdp\beta$. Define $v$ by
$v_i\defeq\infty$ for $b_i\leq p_i$ and $v_i\defeq\max(p_i,a_i)$
otherwise.  Then $\sig{\spair\alpha\gamma}|\sig{\spair\beta\gamma}$ if
and only if $\hdp\gamma|x^v$.
}
\begin{proof}
The assumptions about sig-lead ratios imply that
\[
\sig{\spair\alpha\gamma}=\frac{\sig\gamma}{\gcd(\hdp\alpha,\hdp\gamma)}\sig\alpha
\]
and that
\[
\sig{\spair\beta\gamma}=\frac{\sig\gamma}{\gcd(\hdp\beta,\hdp\gamma)}\sig\beta.
\]
Let $x^q\defeq\frac{\sig\beta}{\sig\alpha}$, $x^c\defeq\hdp\gamma$ and
define $r$ by $r_i=\infty$ for $b_i\leq q_i+a_i$ and $r_i=q_i+a_i$
otherwise. Then what we need to prove for each $i$ is that
\[
\min(b_i,c_i)-\min(a_i,c_i)\leq q_i
\]
if and only if $c\leq\max(a,r)$.

\proofPart{\text{The case }c_i\leq a_i} Both inequalities are always
satisfied in this case since $A|B$ implies that $q_i\geq 0$.

\proofPart{\text{The case }c_i>a_i} We need to prove that
$\min(b_i,c_i)\leq q_i+a_i$ if and only if $c_i\leq r_i$, which follows
quickly from considering the two cases $b_i>q_i$ and $b_i\leq q_i$.
\end{proof}

\subsection{Sig-Lead Ratios}

Section \ref{sec:ratio} shows that sig-lead ratios occur in many
places in \galgex{}. We have wondered if there might be some
mathematical significance to that, but we have not yet found one.

\subsection{Stop on Detecting \grobner{} Basis}
\label{sec:stopgb}

\galg{} sometimes computes a \grobner{} basis much sooner than it
computes a signature \grobner{} basis. So if all we want is a
\grobner{} basis, then we want to stop early.

Call a basis element $\alpha$ \emph{essential} if its lead term
$\hdp\alpha$ is not divisible by the lead term of any other basis
element. Eder, Gash and Perry had the idea to stop the F5 algorithm
early when no S-pair remains that is between two essential basis
elements \cite{f5ExtraTerm}. This idea also works for \galg{} if we
additionally require that at least one S-pair $\spair\beta\gamma$ has
been reduced for each non-essential basis element $\beta$ where
$\hdp\beta|\hdp\gamma$. This is not hard to prove using the classic
Buchberger S-pair elimination criteria. Unfortunately, this criterion
for detecting a \grobner{} basis is not useful for \galg{} --- usually
very little computation can be skipped.

There are sophisticated approaches designed to ensure termination of
F5 using classic \grobner{} basis criteria
\cite{arsThesis,f5ExtraTerm, gashThesis}. We are only concerned with
increasing speed by stopping early, and we have used a very simple
approach. To get an ``if and only if'' criterion for detecting a
\grobner{} basis, we run a classic Buchberger algorithm on the
basis. The classic computation never adds a polynomial to the basis;
when it discovers a lead term that cannot be reduced by the current
basis, it pauses until that lead term can be reduced. This can cause
significant overhead, but it is gua\-ran\-te\-ed to stop the
computation as soon as the basis is a \grobner{} basis.

Using this technique, we see a 30x speed up on yang1 and a 17x
slowdown on katsura11. This technique is only a benefit when the
signature \grobner{} basis is much larger than the minimal \grobner{}
basis, but if we already know that ahead of time, then probably we
should not use \galg{} in the first place. For hcyclic8, the very last
basis element to be added to the signature basis is in fact an
essential basis element. So for hcyclic and examples like it, there is
not much to win from this idea even if there were a zero-overhead way
of doing it.

\section{Data Structures}

We give more background on the material in Section \ref{sec:ds}.

\subsection{Ordering Terms During Reduction}
\label{app:prio}

We are quite surprised at our result that once you apply a hash table,
then it matters little whether you use a geobucket, a heap or a
tournament tree to order the terms. This suggests that there are many
like terms when performing polynomial reduction in \galg{}. If that
hypothesis is correct it explains why hash tables are performing well
in our experiments, since hash tables immediately sum all the like
terms into a single term. If that leaves only a few terms that go into
the priority queue, then that also explains why we are seeing that the
choice of priority queue is less important to the running time. This
is a topic that we want to investigate more closely in future work. An
important topic that we have not yet addressed is what happens when
using packed monomials. The packed monomial technique only applies to
ideals with few enough variables and small enough degree of monomials
in the basis, but many interesting ideals fall into that category.

Some readers may wonder why we bother investigating the classic setup
for polynomial reduction when everyone knows that matrix-based
reduction as in F4 \cite{f4} is much better. We have several
reasons. Priority queues are used throughout the implementation for
several different purposes, so it is important to investigate priority
queues for \grobner{} basis computation regardless. F4 is not a win
for polynomial reduction of a single polynomial, which is an operation
that algorithms outside of \grobner{} basis computation sometimes have
to do.

More importantly, from our conversations with researchers in the
field, we know that implementing F4 to be more efficient than the
classic approach is tricky and that several people have failed in
their attempts. We have even been given the advice to write an F4
implementation such that the reduction steps taken are exactly the
same as those that would be done by classic polynomial reduction, so
that the benefit would accrue only from the replacement of monomials
with column-index-integers. The F4 implementation in Macaulay 2
follows this advice, and as Table \ref{tab:examples} shows, Macaulay 2
gets a very respectable time on hcyclic-8 even though F4 is the only
special thing it does. This indicates to us that it is possible that
better data structures for classic polynomial reduction might make it
just as good as F4 is. We cannot determine as a field if that is true
or not without looking for better data structures for classic
polynomial reduction, and that includes considering better choices of
priority queues. The requirement to store very large matrices in
memory is also a disadvantage of F4.

We give more details on how we have implemented the priority
queues. Most any text book on heaps will explain how to pack a heap
into an array and the basic heap algorithms for inserting and removing
elements. However, we have been unable to find a reference that
collects the various techniques that go beyond the basics. The
literature that mentions the word ``heap'' consists of more than
40,000 articles on Google Scholar, so it seems likely that one of them
explains how to write a good implementation. Yet we have found no such
article, so we collect the improvements to heaps that we know of here,
since these techniques are necessary to be able to replicate our
findings about heaps.

\subsubsection*{Start indices at 1}
If the heap's root is placed at index 0 of the array, then the
formulas for the left child $l(n)$, right child $r(n)$ and parent
$p(n)$ of the node at index $n$ are (division by 2 rounds down)
\[
p(n)=\frac{n-1}{2},\quad
l(n)=2n+1,\quad
r(n)=2n+2.
\]
If the heap's root is placed at index 1, the formulas become instead
\[
p(n)=\frac{n}{2},\quad
l(n)=2n,\quad
r(n)=2n+1.
\]
The latter formulas are more efficient, so place the root at index 1
instead of 0. This can be achieved by leaving the space at index 0
unused, or if using pointers it can be done by subtraxting 1 from the
pointer to the array. Though be aware that subtracting 1 from a
pointer to an array will calculate an invalid pointer. Even merely
calculating an invalid pointer without dereferencing it invokes
undefined behavior according to the C++ standard. However, actual
systems seem to have no problem with it.

\subsubsection*{Make pop move element to bottom of heap before replacement with leaf}

Pop on a heap is frequently described by replacing the top element by
the right-most bottom leaf and then moving it down until the heap
property is restored. This requires 2 comparisons per level that we go
past, and we are likely to go far down the heap since we moved a leaf
to the top of the heap. So we should expect a little less than
$2\log(n)$ comparisons.

Instead, as is a common technique in implementations, we can leave a
hole in the heap where the top element was. Then we move that hole
down the heap by iteratively moving the larger child up. This requires
only 1 comparison per level that we go past. In this way the hole will
become a leaf. At this point we can move the right-most bottom leaf
into the position of the hole and move that value up until the heap
property is restored. Since the value we moved was a leaf we do not
expect it to move very far up the tree. So we should expect a little
more than $\log n$ comparisons, which is better than before.

This technique is widely used in heap implementations, where the
heuristic argument above is borne out in practice through showing
better heap performance.

\subsubsection*{Support replace-top}
Suppose you want to remove the max element and also insert a new
element. Then you can follow the pop algorithm described above, except
that instead of moving the right-most bottom leaf into the vacant
position, you use the new value you wish to push. This is more
efficient than a pop followed by a push.

\subsubsection*{Make the elements have a size in bytes that is a power of 2}
The parent, left-child and right-child formulas work on indices and
they cannot be made to work directly on pointer values. So we are
going to be working with indices and that implies looking up values in
an array $p$ from an index $i$. If $p$ points to an array of $T$s
which each take up $s$ bytes, then the element at offset $i$ in the
array has address $p+s*i$. Using C++ notation, we have that
\[\textrm{
\tt \&(p[i]) == static\_cast<char*>(p) + sizeof(T) * i.
}\]
The computer has to perform this computation to get at the element
with index $i$. $s$ is a compile-time constant, and the multiplication
can be done more efficiently if sizeof(T) is a power of two. We found
an increase in performance by adding 8 padding bytes to increase $s$
from 24 to 32. Reduced efficiency of the cache probably means that
this is not a win for sufficiently large data sets.

\subsubsection*{Pre-multiply indices}
The only thing a heap ever does with an index other than finding
parent, left-child and right-child is to look the index up in an
array. So if we keep track of $j\defeq s * i$ instead of an index $i$,
then we could do a lookup of the element at index $i$ without the
multiplication that would otherwise be necessary since the address of
the element with index $i$ is $p+s*i=p+j$. In C++ notation, we have
that
\begin{align*}
\textrm{\tt
\&(p[i])
}&\textrm{\tt
 == static\_cast<char*>(p) + sizeof(T) * i
}\\&\textrm{\tt
  == static\_cast<char*>(p) + j
}
\end{align*}
Then the left-child and right-child formulas for $j$-values become
respectively $2j$ and $2j + s$. The parent formula becomes
$(j/(2s))*s$ where the division rounds down. In C++ notation, we have
that the parent of $j$ is
\[\textrm{\tt
(j / (2 * sizeof(T))) * sizeof(T).
}\]
If $s$ is a power of 2 then this will compile to 2 bit-shifts. That is
one operation more than the usual parent using indices i. However we
then save one operation on lookup. So the net effect is that finding
the parent takes the same amount of time this way, while lookup of
left-child and right-child becomes faster.

\subsubsection*{Memory optimization}

There are heap-like priority queues that are designed to make better
use of the CPU cache. For example a 4-heap should have fewer cache
misses and the amount of extra overhead is not that much. LaMarca and
Ladner report in 1996 that they get a 75\% performance improvement
from going to aligned 4-heaps \cite{cacheHeap}. However, Hendriks
reports in 2010 that \cite{cacheHeap2}:

\begin{quote}
The improvements to the implicit heap suggested by LaMarca and Ladner
to improve data locality and reduce cache misses were also tested. We
implemented a four-way heap, that indeed shows a slightly better
consistency than the two-way heap for very skewed input data, but only
for very large queue sizes. Very large queue sizes are better handled
by the hierarchical heap.
\end{quote}

Based on this we have chosen not to investigate alternative
cache-friendly heap layouts further.

\subsection{Tournament Trees}

A tournament tree is a classic priority queue data structure that
consists of a binary tree where each node is labeled by a value. Each
interior node's value is the maximum of the values of its two
children. Thus the top of the tree will have the maximum element of
the tree. The data structure is called a tournament tree since such
trees describe for example a tennis tournament where the roots
describe the players and the internal nodes describe a match between
two players. The player at the root of the tree won all his matches
and thus wins the tournament.

Our tournament trees are complete binary trees so that we can pack
them into an array just as is typically done for heaps. We use the
same code to navigate the tree, so the comments for navigating heaps
also apply to our tournament tree implementation. The pre-multiply
indices improvement is not useful, though, since we only ever go up
the tree, and that optimization does not speed up calculating the
parent of a node.

Replacing a value $a$ with another value $b$ is especially fast in a
tournament tree. Every value is annotated with the index of the leaf
node where it comes from, so given any position in the tree, we can
quickly jump to the leaf with the same value. Once the leaf for $a$ is
located, we change its value to $b$ and follow the path from that leaf
to the node while updating the values in the nodes along the way. This
requires only one comparison per level of the tree even in the worst
case.

\subsection{Monomial Ideal Data Structures}
\label{app:monds}

We add detail to Section \ref{sec:monds}.

\subsubsection{Kd-trees}

A leaf in our kd-tree is split into two smaller leaves if it contains
too many monomials. The index of the new internal node is $i+1$ if its
parent has index $i$, starting over at the first variable if there is
no variable $i+1$. We tried choosing the exponent $k$ in $x_{i+1}^k$
to be the median exponent of $x_{i+1}$ among the monomials in the leaf
being split, but we found that it works just as well to let $k$ be the
average of the maximum and minimum exponents of $x_{i+1}$ among the
monomials.

In our implementation we use a special encoding of the binary tree
such that each node on a left-going path down the tree is packed into
a single ``super node''. This technique was very complicated to
implement and gave only a tiny improvement in speed (<5\%), so we
recommend the usual representation of a binary tree where an internal
node has a pointer to each of its children.

The kd-tree can become unbalanced after many insertions and deletions,
especially since we never remove leaves. To combat this, we completely
rebuild the tree after a certain percentage of the tree has changed
since the last rebuild.

\subsubsection{Divmasks}

Table \ref{tab:divmask} shows the hit rates for divmasks in our
implementation of \galg{} when using a monomial list. We say that a
pair of monomials $(a,b)$ is a \emph{divmask hit} if the divmask
proves that $a$ cannot divide $b$. We say that $(a,b)$ is a
\emph{divmask miss} if $a$ does not divide $b$, but the divmask fails
to prove that. We say that $(a,b)$ is \emph{divisible} if $a$ divides
$b$. A divmask can do nothing useful in case of divisibility.  The
\emph{hit rate} for divmasks is the ratio of hits to the sum of hits
and misses.  The \emph{effective hit rate} is the ratio of hits to all
checks for divisibility involving a divmask.

\begin{table*}
\centering
\begin{tabular}{l|rrrrrr}&joswig101&jason210&katsura10&katsura11&hcyclic8&mayr42
  \\
\hline
\# divmask hits&6,347,442,512&674,026,292&81,661,319&703,178,965&19,629,163,403&327,387,283,068 \\
\# divmask misses&1,256,265,766&10,583,467&773,448&4,198,228&204,086,138&2,988,880,796 \\
\#
divisibilities&1,109,093,540&1,667,422&533,541&2,605,217&56,053,435&160,872,762
\\
\hline
hit rate&83.5\%&98.5\%&99.1\%&99.4\%&99.0\%&99.1\% \\
effective hit rate&72.9\%&98.2\%&98.4\%&99.0\%&98.7\%&99.0\% \\

\end{tabular}
\caption{Divmask hit rates}
\label{tab:divmask}
\end{table*}

For most of the examples we get a 99\% hit rate and a 98\% effective
hit rate. This is the case even for mayr42 where there are 51
variables, while our divmasks are 32 bits long so the divmask is
constructed based on the first 32 variables only. joswig101 gets a
lower hit rate of 84\%, but Table \ref{tab:sbTimes} still shows a
large speed up from using divmasks. The lower hit rate is likely due
to the example having only 5 variables.

For the divmap $d_{x_i^t}$ we chooose the exponent $t$ to be the
average of the minimum and maximum exponent of $x_i$ among the
monomials in the data structure. We also tried to use the median, but
there was no advantage.

We use a 32 bit word for the divmasks. If there are less than 32
variables then each variable $x_i$ gets more than one divmap, and we
space the exponents $t$ in the range between the minimum and maximum
exponent of $x_i$. If there are more than 32 variables, then the
variables past the first 32 are ignored when constructing a divmask.

Even if there are more than 32 variables and the program is being
compiled for and run on a 64 bit CPU, it has been slower in our
experiments to use a 64 bit divmask. They do give a higher hit rate on
for example yang1, but the hit rate is already high for 32 bits, so
the increased hit rate did not make up for the extra overhead of
dealing with more bits. We tried 16 bit divmasks too and they were
also worse than 32 bits.

\subsection{The S-pair triangle}

In Section \ref{sec:sptri} we mention a scheme to use 16 bits for
entries in columns in an S-pair triangle when possible and then
switching to 32 bits for columns past $2^{16}$. This idea can be
extended to use $b$ bits from column $2^{b-1}$ to column $2^b-1$. This
scheme is complicated to implement and we calculated that the memory
savings are tiny. We conclude that the split into 16 and 32 bits
already extracts most of the possible benefit, so there is not much
reason to go further.

On yang1 we spend a substantial amount of time on sorting the S-pairs
in each column of the S-pair triangle.  The sorting algorithm we use
to sort each column is the {\tt std::sort} function from the standard
C++ library of GCC. It uses a variant of quicksort. We suspect that
the S-pairs are being constructed in roughly increasing order of
signature, since the basis is ordered by signature, so the array being
sorted should be already in roughly sorted order --- not exactly in
sorted order, but closer to it than a random permutation. Sorting
algorithms that run more quickly on roughly sorted data are called
\emph{adaptive}. Quicksort is still $\Theta(n\log n)$ even when
running on already sorted input. There might be an advantage to be had
from using an adaptive sorting algorithm to sort the columns of the
S-pair triangle.

\section{Our Classic Buchberger Algorithm Implementation}

Our classic Buchberger implementation uses the data structures that we
propose in this paper. Other than that, it has an interesting S-pair
elimination criterion based on an old unpublished theorem of Dave
Bayer. We had originally described classic S-pair elimination criteria
in the paper, but we had to cut it to make the paper fit within the
page limit. What we wrote on the matter is now in this appendix.

\subsection{The lcm and Relatively Prime Criteria}

Many S-polynomials reduce to zero and it is better to predict that and
eliminate the S-pair instead of performing the reduction. Recall that
if all S-pairs have a representation (see \appRef{app:classicDiv}),
then the current polynomial basis is a \grobner{} basis. Buchberger
already introduced two criteria for making that prediction. Let $a$,
$b$ and $c$ be basis elements. If $\hd a$ and $\hd b$ are relatively
prime then $(a,b)$ can be eliminated. This is Buchberger's first
criterion. Call it the \emph{relatively prime criterion}. If $\hd
c|\lcm(\hd a,\hd b)$ and both $\spoly a c$ and $\spoly b c$ have a
representation, then so does $\spoly a b$. This is Buchberger's second
criterion. Call it the \emph{lcm criterion}.

Keeping track of S-pairs takes both time and space in addition to the
time spent on reductions. Therefore it is good to eliminate S-pairs as
early as possible in the computation. So if $\hd c|\lcm(\hd a,\hd b)$
then it is tempting to eliminate $(a,b)$ right away even if $\spoly a
c$ or $\spoly b c$ have not been eliminated yet. This could be done
using the argument that we will eventually process $(a,c)$ and $(b,c)$
so we will eventually eliminate $(a,b)$ so we might as well do it
right away. \emph{Do not believe this argument.} The argument is
incorrect because of the possibility that $\hd a|\lcm(\hd c,\hd
b)$. In that case we would eliminate $(a,b)$ based on an assumption
that we will process $(a,c)$ later, and we would also eliminate
$(a,c)$ based on an assumption that we will process $(a,b)$
later.

There is a revised approach which does work. Assume that $\hd
c|\lcm(\hd a,\hd b)$. Then we can eliminate $(a,b)$ right away if
\[
\lcm(a,c)\neq\lcm(a,b)\text{ or }(a,c)\text{ is eliminated,}
\]
and
\[
\lcm(b,c)\neq\lcm(a,b)\text{ or }(b,c)\text{ is eliminated.}
\]
Here an S-pair is also considered to be eliminated if it has been
reduced. In this case there is no possibility of a circular
argument. There are many alternatives to this particular way of
applying the lcm criterion. Suppose $\lcm(\hd a,\hd b)=\lcm(\hd a,\hd
c)=\lcm(\hd b, \hd c)$. Then one and only one of the S-pairs among
$a,b,c$ can be eliminated. The approach outlined here in effect lets
the ordering on the S-pairs decide - the S-pair ordered to be reduced
last is the one to get eliminated. If the S-pair ordering is good then
the last pair should also be the pair that we would most like to avoid
reducing, so this is a good choice. Furthermore, this scheme is simple
to think about and simple to implement correctly.

Using the lcm criterion in this way requires a way to quickly
determine if a given S-pair $(a,b)$ has been eliminated. We solved
this problem by keeping a triangular array of $\binom{k}{2}$ bits (not
bytes) which supports constant time access to a bit for each S-pair.

The monomial data structures from Section \ref{sec:monds} can be used
to find the divisors of $\lcm(\hd a,\hd b)$. However, there can be so
many S-pairs that the lcm criterion becomes a bottleneck even with
those data structures. To get around this, we cache the elements of
$B$ that often end up as the $c$ such that $\hd c|\lcm(\hd a,\hd
b)$. If $c$ is the element that eliminates $(a,b)$ using the lcm
criterion, then we associate $c$ to both of $a$ and $b$, and we forget
any earlier such association to $a$ and $b$. The next time we consider
an S-pair involving $a$, we see that $c$ was useful before and check
$c$ first before performing a full search of all divisors. Given a
pair $(a,b)$, there will be two cached generators to check --- one for
$a$ and one for $b$. We were surprised by the effectiveness of this
approach --- see \appRef{app:eval}.

\subsection{The Graph Criterion}

The lcm criterion has been improved on by Gebauer and M\"oller
\cite{gebMolCrit}. They propose a technique that quickly constructs a
near-minimal set of S-pairs.  It is well-known that in order to obtain
a \grobner{} basis, it suffices to reduce only those S-pairs which
correspond to a minimal generating set of the syzygy module of the
lead terms of the basis.  Caboara, Kreuzer and Robbiano show that
there is an advantage to be had by getting this minimal set instead of
relying on heuristics \cite{minSyzGB}. They propose an algorithm to
obtain the minimal set which in some cases leads to a speed up. There
is also an overhead to the computation so that it does not always pay
off.

In the 1980's, Dave Bayer came up with an unpublished alternative
graph-based characterization of the minimal set of S-pairs that form a
generating set. We use it to obtain the minimal set without much
overhead. Given a polynomial basis $B$ and a monomial $m$, define an
undirected graph $G_m$ with vertices
\[
\setBuilder{g\in B}{\hd g\text{ divides }m}
\]
such that $(a,b)$ is an edge if $\lcm(\hd a,\hd b)\neq m$ or if
$(a,b)$ has already been eliminated. We eliminate an S-pair $(a,b)$ if
$a$ and $b$ are connected in $G_{\lcm(a,b)}$. The set of S-pairs that
are not eliminated then correspond to a minimal \grobner{} basis of
the set of syzygies on the initial terms of the basis. Call this
criterion the \emph{graph criterion}. Observe that the lcm criterion
as described above is a special case of this more powerful criterion,
but that the lcm criterion has less overhead. We use both criteria in
our implementation.

If there is a cycle in $G_m$, then any edge $(a,b)$ on the cycle with
$\lcm(\hd a,\hd b)=m$ can be eliminated. As before, we in effect let
the ordering on the S-pairs choose which S-pair it would least like to
do.

In practical terms we implement this criterion by computing the graph
$G_{\lcm(a,b)}$ for each S-pair $(a,b)$ just before $\spoly a b$ would
otherwise have been reduced. The nodes of the graph can be determined
quickly using the monomial data structures from Section
\ref{sec:monds}. The graphs are usually small and the edges are quick
to construct, so the overhead is not much. Especially not since the
early and late lcm criterion and the relatively prime criterion
already eliminate most of the S-pairs. It should be possible to make
significant optimizations by caching parts of the graph, but we have
seen no need to improve this part of our implementation as it does not
take much time.

\subsection{Evaluating S-pair Elimination Criteria}
\label{app:eval}

Table \ref{tab:prune} shows how many S-pairs are eliminated by each
criterion. Every row in this table shows something interesting. The
relatively prime criterion is very effective on yang1, eliminating
52\% of the S-pairs. For yang1, the cache idea works so well that 98\%
of the S-pairs that the lcm criterion can eliminate are eliminated
already from just looking at the cache. The graph criterion gets only
an extra 4\% eliminated S-pairs compared to not using it on yang1, and
on 4by4 the number is down to 0.2\%.

This business of counting the number of S-pairs that are eliminated is
a good first approximation, but ultimately what counts is the time
saved due to the eliminated S-pairs. In some cases a small number of
zero reductions can account for a large amount of the running time of
the algorithm, so it can be the case that the graph criterion can
eliminate just a few S-pairs and still contribute a significant speed
up.

Another point is that it is not very informative to compare the number
of zero reductions between Table \ref{tab:prune} and Table
\ref{tab:sigSPair}. The overhead from \galgex{} is not just in
reducing to zero, it is also in reducing to a basis element that is
part of the signature \grobner{} basis but not part of the minimal
\grobner{} basis. A better comparison would be in terms of total
number of divisor queries, monomial multiplications, monomial
comparisons and ground field operations. Even that is not perfect, but
it is better than just looking at the number of reductions performed,
let alone just looking at the number of reductions to zero. This is a
kind of measure that we are looking into having our implementation be
able to report.

\begin{table}
\centering
\begin{tabular}{lrrrrr}
Example&4by4&jason210&mayr42&yang1 \\
\hline
\#S-pairs  & 108,811 & 404,550 & 36,410,311 & 11,331,180\\
rel prime    & 27,683 & 2 & 664,223 & 5,426,819 \\
lcm cache hits      & 50,051 & 353,112 & 30,995,451 & 5,643,927 \\
lcm simple hits     & 26,182 & 45,604 & 4,411,353 & 162,472 \\
lcm graph hits      & 12 & 81 & 5,363 & 3,964 \\
\hline
\#reductions & 4883 & 5751 & 333,921 & 93,998 \\
0-reductions & 4416 & 4851 & 325,387 & 89,237\\
\end{tabular}
\caption{Classic Buchberger S-pair elimination}
\label{tab:prune}
\end{table}
}

\bibliographystyle{abbrv}
\bibliography{references}
\end{document}